\documentclass[11pt,letterpaper]{article}

\usepackage[margin=1in]{geometry}

\usepackage{lineno}

\usepackage{cite}
\usepackage{hyperref}
\usepackage{colortbl}
\usepackage{paralist}
\usepackage{enumerate}
\usepackage{pdfpages}
\usepackage{float}
\usepackage{booktabs}
\usepackage[override]{cmtt}
\usepackage{color,xcolor}
\usepackage{graphicx,eso-pic}
\usepackage{boxedminipage}
\usepackage{url}
\usepackage{caption}
\usepackage{subcaption}
\usepackage{xspace}
\usepackage{multirow}
\usepackage{amsmath,amsthm,amstext,amssymb,amsfonts,latexsym}
\usepackage{wrapfig}
\usepackage{algorithm}
\usepackage{algpseudocode}
\usepackage{algorithmicx}
\usepackage{epstopdf}
\usepackage{mdframed}
\usepackage{cases}
\usepackage[capitalize]{cleveref}
\usepackage[compact]{titlesec}
\usepackage{bm}

\newcounter{note}[section]
\newcommand{\hubert}[1]{\refstepcounter{note}$\ll${\sf Hubert'sComment~\thenote:} {\sf \textcolor{blue}{#1}}$\gg$\marginpar{\tiny\bf HC~\thenote}}

\newcommand{\yuetodo}[1]{{\large\color{green}[Yue todo: #1]}}

\definecolor{yue}{rgb}{0.7, 0, 0}

\renewcommand{\yuetodo}[1]{}

\newcommand{\mcal}[1]{\ensuremath{\mathcal {#1}}}

\definecolor{darkgreen}{rgb}{0,0.5,0}
\definecolor{lightblue}{RGB}{0,176,240}
\definecolor{darkblue}{RGB}{0,112,192}
\definecolor{lightpurple}{RGB}{124, 66, 168}
\definecolor{grey}{RGB}{139, 137, 137}
\definecolor{maroon}{RGB}{178, 34, 34}
\definecolor{green}{RGB}{34, 139, 34}
\definecolor{types}{RGB}{72, 61, 139}
\definecolor{gold}{rgb}{0.8, 0.33, 0.0}

\definecolor{darkgray}{gray}{0.3}

\definecolor{darkred}{rgb}{0.5, 0, 0}
\definecolor{darkgreen}{rgb}{0, 0.5, 0}
\definecolor{darkblue}{rgb}{0,0,0.5}

\newcommand\markx[2]{}

\newcommand{\R}{\mathbb{R}}

\newcommand{\Z}{\mathbb{Z}}

\newcommand{\ignore}[1]{}

\newcounter{task}

\newenvironment{proofof}[1]{\begin{proof}[Proof of #1]}{\end{proof}}
\newtheorem{theorem}{Theorem}[section]

\newtheorem{corollary}[theorem]{Corollary}
\newtheorem{fact}[theorem]{Fact}

\newtheorem{lemma}[theorem]{Lemma}

{
\theoremstyle{definition}
\newtheorem{definition}[theorem]{Definition}
\newtheorem{remark}[theorem]{Remark}

\newtheorem{assumption}[theorem]{Assumption}
}

\newcounter{cnt:challenge}

\newcommand*\samethanks[1][\value{footnote}]{\footnotemark[#1]}

\newcommand{\D}{\mathsf{D}}

\newcommand{\util}{\mathsf{u}}

\newcommand{\f}{\mathfrak{f}}
\newcommand{\g}{\mathfrak{g}}

\newcommand{\Pow}{\mathsf{Pow}}

\DeclareMathOperator*{\argmin}{arg\,min}
\DeclareMathOperator*{\argmax}{arg\,max}

\newcommand{\vecp}{\mathbf{p}}
\newcommand{\vecq}{\mathbf{q}}

\newcommand{\Bfg}{\mcal{B}^{\geq}_\f \times \mcal{B}^{\leq}_\g}

\title{Density Decomposition in Dual-Modular Optimization: \\
Markets, Fairness, and Contracts}

\author{T-H. Hubert Chan\thanks{Department of Computer Science, the University of Hong Kong.} \and Shinuo Ma\samethanks }

\date{}

\begin{document}

\begin{titlepage}
    \maketitle

    \begin{abstract}

We study a unified framework for optimization problems defined on dual-modular instances, where the input comprises a finite ground set $V$ and two set functions: a monotone supermodular reward function $\f$ and a strictly monotone submodular cost function $\g$.  This abstraction captures and generalizes classical models in economics and combinatorial optimization, including submodular utility allocation (SUA) markets and combinatorial contracts. At the core of our framework is the notion of density decomposition, which extends classical results to the dual-modular setting and uncovers structural insights into fairness and optimality.

We show that the density decomposition yields a canonical vector of reward-to-cost ratios (densities) that simultaneously characterizes market equilibria, fair allocations -- via both lexicographic optimality and local maximin conditions -- and best-response strategies in contract design. Our main result proves the equivalence of these fairness notions and guarantees the existence of allocations that realize the decomposition densities.

Our technical contributions include the analysis of a broad family of convex programs -- parameterized by divergences such as quadratic, logarithmic, and hockey-stick functions -- whose minimizers recover the density decomposition. We prove that any strictly convex divergence yields the same canonical density vector, and that locally maximin allocations act as universal minimizers for all divergences satisfying the data processing inequality.

As an application of our framework, we determine the structure and number of critical values in the combinatorial contracts problem. Additionally, we generalize a Frank-Wolfe-type iterative method for approximating the dual-modular density decomposition, establishing both convergence guarantees and practical potential through efficient gradient oracle design.

Our framework unifies and generalizes several lines of prior work and provides a principled foundation for studying density-based optimization in dual-modular and economic settings.

\end{abstract}


    \thispagestyle{empty}
\end{titlepage}

\section{Introduction}
\label{sec:intro}

In this work, we consider problems whose input instance
$(V; \f, \g)$
consists of a finite ground set~$V$ (where $n = |V|$) and
two set functions $\f, \g: 2^V \rightarrow \R_{\geq 0}$ satisfying
a \emph{dual-modular} requirement. Specifically, the \emph{reward} function $\f$ is supermodular and monotone,
and the \emph{cost} function $\g$ is  submodular and (strictly\footnote{We shall
see in Section~\ref{sec:strict_example} why strict monotonicity is required for $\g$.}) monotone.
The input instance is interpreted in the following economic problems
that we shall study in a unified framework.  We will describe the definitions
in an intuitive way in the introduction, while the fine technical details
are given in the main body.

\noindent \textbf{Dual-Modular Markets and Polymatroids.}  This is a generalization
of the \emph{submodular utility allocation market} (SUA) proposed by Jain and 
Vazirani~\cite{DBLP:journals/geb/JainV10},
which can be equivalently formulated by a supermodular reward function~$\f$
and the special case of a linear\footnote{A set function $\g: 2^V \rightarrow \R_{\geq 0}$
is linear if there exists a vector $y \in \R^V$
such that $\g(S) = \sum_{v \in S} y_v$; we also use $y(v) = y_v$ to
denote a coordinate of the vector.} cost function~$\g$ 
(which, as we shall see, is equivalent to assigning 
a fixed quantity $y_v = \g(\{v\})$ to each agent~$v$).

Here, the ground set~$V$ represents a collection of \emph{agents}
and there are total $\f(V)$ amount of reward and $\g(V)$ amount of cost,
both of which need to be distributed among the agents
such that each agent receives non-negative amounts of reward and cost.
Specifically, an \emph{allocation} is a pair $(x, y) \in \R_{\geq 0}^V \times \R_{\geq 0}^V$
of reward and cost allocation vectors that satisfy the following requirements:

\begin{compactitem}

\item  For all $S \subseteq V$, $\sum_{v \in S} x_v \geq \f(S)$;
this guarantees the minimum reward for the group $S$ of agents.
In addition,
the coordinates of $x$ sum to $\f(V)$.

Equivalently, $x$ is in
the \emph{base contrapolymatroid}~$\mcal{B}^{\geq}_\f$.

\item  For all $S \subseteq V$, $\sum_{v \in S} y_v \leq \g(S)$;
this puts a limit on the maximum cost charged to the group~$S$.
In addition, the coordinates of $y$ sum to $\g(V)$.

Equivalently, $y$ is in the \emph{base polymatroid}~$\mcal{B}^{\leq}_\g$.
\end{compactitem}

\noindent \emph{Market Equilibrium.}  The SUA market itself
is a generalization of the Fisher market~\cite{fisher1892}.
In the literature, market equilibrium can be characterized by 
several approaches.

\begin{compactitem}

\item \emph{Convex Program.}
An SUA market equilibrium is defined to be an optimal solution
of some Eisenberg-Gale-type convex program in~\cite{DBLP:journals/geb/JainV10},
but it is difficult to understand the intuition of the equilibrium
with this characterization.
Moreover,
there is actually nothing specific about the Eisenberg-Gale objective function and it is known~\cite{DBLP:conf/ipco/Nagano07,Panditi2016} 
that the SUA market equilibrium can be characterized by a
wide class of convex objective functions including the one
discovered by Shymrev~\cite{shmyrev2009}.

\item \emph{Fairness Notion.}
A market equilibrium may also be characterized by
fairness conditions 
based on the satisfaction of the agents.
In the context of dual-modular markets,
the satisfaction of an agent~$v \in V$ in an allocation $(x, y)$
is measured by the density $\rho_v := \frac{x_v}{y_v}$,
which is also known as the \emph{reward-to-cost} or
\emph{bang-to-buck} ratio. 

\emph{Lexicographic optimality}~\cite{fujishige1980lexicographically} takes a global view
on fairness.  First, consider the collection~$\mcal{A}_1$ of allocations in which the maximum density among all agents is minimized.
For $i \geq 1$, suppose we have formed the collection $\mcal{A}_i$ of allocations,
and in every such allocation, the largest~$i$ densities among the agents
form the same multi-set.
Then, $\mcal{A}_{i+1} \subseteq \mcal{A}_i$ contains those allocations
in which the $(i+1)$-st largest density is minimized.
Eventually $\mcal{A}_n$ is the collection of lexicographically optimal allocations.\footnote{In case the reader wonders,
lexicographic optimality may equivalently be defined in the complementary approach by 
first maximizing the minimum density among all agents. Keep reading until the end of the introduction
to see a proof.}

\emph{Locally maximin condition}~\cite{DBLP:conf/innovations/ChanX25} gives a way to check whether a given allocation is fair.
Loosely speaking,
the condition states that if an agent~$i$ is less
satisfied than agent~$j$ (because $\rho_i < \rho_j$) in
an allocation, then the allocation already \emph{prioritizes}
agent~$i$ over agent~$j$.

Formally, suppose in an allocation~$(x, y)$,
for $\rho > 0$, $S^{(\rho)}$ is the collection of agents~$v$
whose induced density $\rho_v = \frac{x_v}{y_v}$ is at least $\rho$.
Then, the locally maximin condition requires that for any $\rho > 0$,
the agents in $S^{(\rho)}$ together receive the worst possible reward and cost,
i.e, $\sum_{v \in S^{(\rho)}} x_v = \f(S^{(\rho)})$
and $\sum_{v \in S^{(\rho)}} y_v = \g(S^{(\rho)})$.

As we shall see, the two fairness notions can be proved
to be equivalent.  However, just from the definitions,
the existence of a lexicographically optimal allocation is
guaranteed, while checking whether an allocation is
locally maximin is convenient.

\end{compactitem}

It is known~\cite{DBLP:conf/ipco/Nagano07,DBLP:conf/innovations/ChanX25}
that these approaches are equivalent for some special cases (e.g., $\g$ is linear)
and they give rise to the same density vector $\rho^* \in \R^V$.
We shall extend these results to the more general dual-modular markets.

\noindent \textbf{Combinatorial Contracts~\cite{DBLP:conf/focs/DuttingEFK21}.} A \emph{principal} would
like to hire an \emph{agent} to perform a subset of actions
from some action set~$V$.  The principal specifies a contract parameter $\alpha \in [0, 1]$ such that if the agent performs $S \subseteq V$,
the agent will receive a utility $\util_a(\alpha, S) := \alpha \cdot \f(S) - \g(S)$,
while the principal will receive a utility $\util_p(\alpha, S) := (1 - \alpha) \cdot \f(S)$.

For each~$\alpha \in [0, 1]$,
the agent has some best response~$S_\alpha \in \arg\max_{S \subseteq V} \util_a(\alpha, S)$, where ties are resolved following two
guidelines. The first requirement is clear: the subset benefiting the principal more will be chosen.

The second guideline is more subtle.  Suppose that when the $S_\alpha$'s are chosen, the collection $\Lambda \subseteq [0, 1]$ of values~$\alpha$
at which $S_\alpha$ changes are known as the \emph{critical values}.
Since the principal would finally like to choose $\alpha \in [0, 1]$ to
maximize $\util_p(\alpha, S_\alpha)$,
it suffices to check the utility for each critical $\alpha \in \Lambda$.
Therefore,
the second guideline for resolving ties when picking all the $S_\alpha$'s
is that, as
 $\alpha$ increases continuously from $0$ to $1$,
the best response $S_\alpha$ will change ``as infrequently as possible''
such that the size of $\Lambda$ would be kept as small as possible.

It is known~\cite{DBLP:conf/focs/DuttingEFK21} that without the dual-modular requirements on $(V; \f, \g)$,
even computing $S_\alpha$ for a specific $\alpha$ is NP-hard,
and there could be exponential number of critical values.
However, it is subsequently shown~\cite{DBLP:conf/soda/DuttingFT24} that
for the dual-modular instances, there are efficient algorithms
to compute a best response $S_\alpha$ for each $\alpha$ and the resulting number
of critical values is at most $|V|$.
Our unified framework will reveal the structure of these best responses
and the critical values.

\subsection{Dual-Modular Density Decomposition}

We shall see that the above problems can be tackled by
the \emph{density decomposition}.
For the special case of linear $\g$,
Fujishige~\cite{fujishige1980lexicographically}
first considered it in the context
of lexicographically optimal base of polymatroids.
This decomposition has been independently rediscovered many times
in the computer science community.
For instance, a well-known special case has been studied~\cite{DBLP:conf/www/TattiG15}, where the ground set is the vertex set in a graph $G= (V, E)$.
Here, given $S \subset V$, $\f(S)$ is the number of edges such that both endpoints 
are in $S$ and $\g(S) = |S|$. The description of the decomposition
extends readily to the dual-modular setting from the special cases~\cite{DBLP:conf/www/TattiG15,DBLP:conf/www/DanischCS17,DBLP:conf/esa/HarbQC23}

The density of a subset $S \subseteq V$ is
$\rho(S) := \frac{\f(S)}{\g(S)}$.
A \emph{densest subset}~$S$ is one that maximizes $\rho(S)$.
Even though there can be more than one densest subset,
for a dual-modular instance,
it is known that the \emph{maximal} densest subset~$S_1$ is unique
and includes all other densest subsets.
The maximal densest subset~$S_1$ is the first step 
of the density decomposition,
and suppose $\rho_1 = \frac{\f(S_1)}{\g(S_1)}$ is the corresponding density.

\noindent \textbf{Connection to Dual-Modular Markets.}
The rules state that the agents in $S_1$ together will receive
a reward of at least $\f(S_1)$ and pay a cost of at most $\g(S_1)$.
Hence, if in an allocation $(x, y)$,
the agents in $S_1$ totally receive and pay exactly these worst-case amounts of reward and cost,
other agents outside $S_1$ would be in the best situation.
One of our main results is that the rules indeed allow the reward amount $\f(S_1)$ and the cost
amount $\g(S_1)$ to be distributed proportionally among $S_1$
such that the induced density of every agent in $S_1$ is exactly $\rho_1$.
If it is possible to allocate the remaining reward
and cost among other agents such that eventually $\rho_1$ turns out to be the maximum density among
all agents in~$V$, then we know that the locally maximin condition
cannot be violated due to an agent in $S_1$;
moreover, in any lexicographically optimal allocation,
all agents $v \in S_1$ must have density $\rho^*_v := \rho_1$.

To continue with the decomposition,
we recursively apply the process to the subinstance
$(V \setminus S_1; \f(\cdot | S_1), \g( \cdot | S_1))$,
where the marginal contribution
is $h(A | S) := h(A \cup S) - h(S)$.
The process terminates when an empty ground set is reached.
If the process terminates in $k$ steps,
then the decomposition gives a partition $V = \cup_{i=1}^k S_i$
on the ground set and a corresponding density vector $\rho^* \in \R^V$.
Specifically, for each $i$, each agent $v \in S_i$
will be assigned some density $\rho^*_v := \rho_i = \rho(S_i | S_{<i})$
in the subinstance in which it is created.
Our first main result connects 
the density decomposition and the dual-modular market fairness.

\begin{theorem}[Dual-Modular Density Decomposition and Market Fairness]
\label{th:main_market}
In a dual-modular instance $(V; \f, \g)$, suppose $\rho^* \in \R^V$ is the 
density vector obtained from the density decomposition.
Then, the following three conditions are equivalent
for an allocation~$(x, y) \in \mcal{B}^{\geq}_\f \times \mcal{B}^{\leq}_\g$:

\begin{compactitem}
\item[(i)] For all $v \in V$,
the induced density $\rho_v := \frac{x_v}{y_v}$ equals $\rho^*_v$.

\item[(ii)] The allocation $(x, y)$ is lexicographically optimal.

\item[(iii)]  The allocation $(x, y)$ is locally maximin.

\end{compactitem}

\end{theorem}

We shall prove a series of structural lemmas
from which we can draw Theorem~\ref{th:main_market}
as a conclusion.  As aforementioned,  it is not clear \emph{a priori}
whether there exists an allocation satisfying condition~(i) or~(iii).
However, a lexicographically optimal allocation must exist by definition,
because the feasible space $\mcal{B}^{\geq}_\f \times \mcal{B}^{\leq}_\g$ is compact.

\noindent \textbf{Revealing Optimal Contract Response with
Density Decomposition.}  We next show how the density decomposition
is connected to the dual-modular contracts problem.
Recall that given the contract parameter $\alpha \in [0, 1]$,
the associated utility
of an agent is $\util_a(\alpha, S) = \alpha \cdot \f(S) - \g(S)
= \alpha \cdot (\f(S) - \frac{1}{\alpha} \cdot \g(S))$.
Hence, intuitively, elements with densities
above the threshold $\frac{1}{\alpha}$ would cause positive utility,
while those with densities below the threshold would cause negative utility.
This intuition can indeed be formalized in the following result.

\begin{theorem}[Density Decomposition Reveals Optimal Contract Response]
\label{th:main_contracts}
Suppose in the density decomposition $V = \cup_{i=1}^k S_i$,
each component $S_i$ is associated with the density $\rho_i$,
where $\rho_1 > \rho_2 > \cdots > \rho_k$.

Given the contract parameter~$\alpha \in [0,1]$,
suppose $i$ is the largest index such that $\rho_i \geq \frac{1}{\alpha}$.
Then, the subset $S_{\leq i} := \cup_{j=1}^i S_j \in \arg \max_{S \subseteq V} \util_a(\alpha, S)$
is an optimal response.

Moreover, if $\rho_i > \frac{1}{\alpha} > \rho_{i+1}$,
then the only optimal response is $S_{\leq i}$.
This implies that the collection of critical values
with the smallest size must be $\{\frac{1}{\rho_i}: 1 \leq i \leq k\} \cap [0,1]$.
\end{theorem}

\subsection{Analyzing Dual-Modular Density Decomposition, Markets
and Contracts with Convex Program}

To establish Theorem~\ref{th:main_market},
we will analyze the density decomposition with convex programs.
This approach has first been proposed by
Fujishige~\cite{fujishige1980lexicographically} for the case
of linear $\g$, and have been analyzed 
extensively~\cite{DBLP:conf/ipco/Nagano07,Panditi2016,DBLP:conf/nips/HarbQC22}.
For the special case
of linear $\g$, we consider a fixed vector $y \in \R^V$ with  $y_v = \g(\{v\})$.

The objective functions can be described as follows.
Given a convex function~$\vartheta: \R_{\geq 0} \rightarrow \R$,
the $\vartheta$-divergence between two
vectors $x, y \in \R^V_{\geq 0}$ is denoted as:

$$\D_\vartheta(x \| y) := \sum_{u \in V} y_u \cdot \vartheta(\frac{x_u}{y_u}).$$

\noindent \textbf{Convex Program.}  One can pick
a convex function~$\vartheta$ and form a convex program as follows.
The domain of the variable $x$ is $\mcal{B}^{\geq}_{\f}$,
and the objective function to be minimized is $\Phi_\vartheta(x) := \D_\vartheta(x \| y)$,
recalling that $y$ is fixed for linear $\g$.
Many works in the literature have analyzed convex programs that
fall under this framework with different choices of
convex $\vartheta$.

\begin{compactitem}

\item Fujishige~\cite{fujishige1980lexicographically}
and subsequent works on graph density decomposition~\cite{DBLP:conf/www/DanischCS17,DBLP:conf/esa/HarbQC23} picked $\vartheta(t) := t^2$.

\item The Eisenberg-Gale convex program and subsequent works on SUA markets~\cite{DBLP:journals/geb/JainV10}
picked $\vartheta(t) := - \log t$.

\item The Shymrev convex program for linear exchange markets~\cite{shmyrev2009} picked $\vartheta(t) := t \log t$.
\end{compactitem}

The following surprising result shows that any strictly convex\footnote{
A function $\vartheta: \R_{\geq 0} \to \R$ is
strictly convex if
 for $0 \leq a < b$
and $0 < \lambda < 1$,
$\vartheta(\lambda a + (1-\lambda) b) < \lambda \vartheta(a) 
+ (1 - \lambda) \vartheta(b)$.
}
 $\vartheta$ may be selected
and the optimal solution is related to the density decomposition.

\begin{fact}[Convex Program and Density Decomposition~\cite{DBLP:conf/ipco/Nagano07,Panditi2016}]
\label{fact:convex_program}
Given an instance $(V; \f, \g)$ with supermodular $\f$ and linear $\g$ (corresponding
to a fixed $y \in \R^V_{>0}$), for any strictly convex $\vartheta$,
the objective function $\Phi_\vartheta(x) := \D_\vartheta(x \| y)$
has exactly the same minimizer $x \in \mcal{B}^{\geq}_{\f}$,
which is the unique lexicographically optimal allocation (with respect to densities).

Moreover, for each $v \in V$, $\frac{x_v}{y_v} = \rho^*_v$,
where $\rho^* \in \R^V$ is the density vector from the density decomposition.
\end{fact}

\noindent \emph{Dual-Modular Instance.}
Given a dual-modular instance $(V; \f, \g)$ and
a convex function~$\vartheta$,  it is natural to consider
the objective function $\Phi_\vartheta(x, y) := \D_\vartheta(x \| y)$
with the domain $(x, y) \in \mcal{B}^{\geq}_\f \times \mcal{B}^{\leq}_\g$.
Although the same objective function is used, for general submodular $\g$,
the variable $y$ is no longer fixed, which would complicate the analysis.
Moreover, in the dual-modular setting, a minimizer of $\Phi_\vartheta$ may not be unique, but
the following result together with Theorem~\ref{th:main_market}
shows that any minimizer would induce the density vector~$\rho^*$
from the density decomposition.  Furthermore, it also implies
the existence of a locally maximin allocation.

\begin{theorem}[Equivalence Between Convex Program and Local Maximin Condition]
\label{th:main_maximin}
Given a dual-modular instance $(V; \f, \g)$ and
a strictly convex function~$\vartheta$,
an allocation $(x, y) \in \mcal{B}^{\geq}_\f \times \mcal{B}^{\leq}_\g$
is a minimizer of $\Phi_\vartheta(x, y) := \D_\vartheta(x \| y)$
\emph{iff} it is locally maximin.

If the convex $\vartheta$ does not 
satisfy strict convexity,
we still have the direction that any local maximin allocation
is a minimizer.
\end{theorem}

\noindent \emph{Augmenting the Domain with Permutations.}
One strategy towards proving Theorem~\ref{th:main_maximin}
is to first show a weaker statement that there exists a locally maximin allocation that is a minimizer of the convex program. To prove this weaker statement, one can start with any minimizer $(x, y)$ of the convex program.
A general idea is to construct a sequence of ``small modifications''
to the solution such that each modification brings it closer
to being locally maximin but without increasing the objective value.
As opposed to the case of linear $\g$, the challenge here is that we need to modify both $x$ and $y$, and we will also need a notion of progress
towards the locally maximin condition.

In order to define what a small modification is,
we need the \emph{permutation} interpretation
of the bases $\mcal{B}^{\geq}_\f$ and $\mcal{B}^{\leq}_\g$.
Given a permutation $\sigma$ of~$V$,
we use $u \prec_\sigma v$ to mean $u$ appears earlier than $v$.
Given a set function $h: 2^V \to \R$,
we can imagine the agents in $V$ arrive one-by-one
according to the permutation and
the value received by $u$ is its marginal contribution
of $h$ given the agents arriving before it.
In other words, we denote $h^\sigma \in \R^V$ as the vector
defined by: 

$$h^\sigma(u) := h(\{u\} | \{w\in V: w\prec_{\sigma} u\}), \forall u \in V.$$

Because of supermodular $\f$ and submodular $\g$,
an agent would prefer being near the end of the permutation to
receive a higher reward and a lower cost.

\begin{fact}[~\cite{SFMCP}]
\label{fact:perm_polytope}
The base $\mcal{B}^{\geq}_\f$ is the convex
hull of $\f^\sigma$ over all permutations~$\sigma$ of $V$;
similarly, 
$\mcal{B}^{\leq}_\g$ is the convex hull
of $\g^\sigma$ over all permutations~$\sigma$.
\end{fact}

In view of Fact~\ref{fact:perm_polytope},
each allocation $(x, y) \in \mcal{B}^{\geq}_\f \times \mcal{B}^{\leq}_\g$
can be represented by a pair $(\vecp, \vecq)$
of permutation distributions,
and we can consider a convex program with the variables 
$(\vecp, \vecq)$ instead.

Note that it is easy to make a small modification to a permutation~$\sigma$.
Suppose there are two adjacent agents $u$ and $v$ in the permutation
such that the $v$ is strictly more satisfied than $u$ in terms of density, but
comes after $u$ (which is a more preferred spot) in the permutation.
To achieve better fairness,
one would like to swap their positions to get a different permutation~$\sigma'$.
When applied to a permutation distribution $\vecp$,
this means transferring all the weight from coordinate~$\sigma$ to
coordinate~$\sigma'$. 

We shall give a formal definition of local maximin for $(\vecp, \vecq)$.
Indeed, Lemma~\ref{lemma:maximin_opt} uses a sequence
of such small modifications to move the solution closer to being locally maximin,
but without increasing the objective function in each step.
The notion of progress towards the local maximin condition resembles swapping adjacent elements in sorting algorithm analysis to reduce the number of inversions.
Even though the number of variables increases
from $2n$ to $2(n!)$, the purpose of the augmented domain is for
analysis only.

\noindent \textbf{Dual-Modular Contracts and Hockey-Stick Divergences.}
Writing $\gamma = \frac{1}{\alpha} \geq 1$ in the contracts problem,
the best response for an agent is equivalent
to the problem: 

$\max\{ \f(S) - \gamma \cdot \g(S): S \subseteq V\}$.

This turns out to a dual problem for the convex program
with $\vartheta_\gamma(t) := \max\{t - \gamma, 0\}$.
Indeed, $\mathsf{HS}_\gamma(P \| Q) := \D_{\vartheta_\gamma}(P \| Q)
= \sup_{S \subseteq \Omega} P(S) - \gamma \cdot Q(S)$
is the \emph{Hockey-Stick} divergence.
Theorem~\ref{th:main_contracts} is a direct consequence of the following result.

\begin{lemma}[Duality Between Contracts and Hockey-Stick]
\label{lemma:main_dual_contracts}
For $\gamma \geq 0$, any $S \subseteq V$ and $(x, y) \in \Bfg$, we have:
$\f(S) - \gamma \cdot \g(S) \leq \mathsf{HS}_\gamma(x \| y)$.

Moreover, equality is attained by the subset constructed in Theorem~\ref{th:main_contracts} and any locally maximin~$(x, y)$.
\end{lemma}

\noindent \textbf{Locally Maximin Allocation as a Universal Minimizer.}
Theorem~\ref{th:main_market} shows that any locally maximin $(x, y)$
is a minimizer of $\D_\vartheta(x \| y)$ for all convex~$\vartheta$.
As a bonus, the argument
in~\cite[Theorem 5.12]{DBLP:conf/innovations/ChanX25}
shows that as long as $(x, y)$ is a minimizer of
$\mathsf{HS}_\gamma(x \| y)$ for all $\gamma \geq 0$,
it will be a universal minimizer for a much wider class
of divergences satisfying the \emph{data processing inequality}
(which is formalized in Definition~\ref{defn:dpi}).

\begin{theorem}[Universal Minimizer for Data-Processing Divergences]
\label{th:main_universal}
Any locally maximin allocation $(x, y) \in \Bfg$
is a minimizer of $\D(x \| y)$
for all divergences $\D$ satisfying the data processing inequality.
\end{theorem}

\subsection{Iterative Approximation Methods}

\noindent \emph{Exact Decomposition vs Approximation.}
As mentioned in~~\cite{DBLP:conf/nips/HarbQC22},
the generalized densest subset problem can be solved in polynomial time
by submodular function minimization (e.g.,
with running time~$O(n^6)$~\cite{DBLP:journals/mp/Orlin09}).
Therefore, the exact density decomposition in Definition~\ref{defn:density_decomp}
can also be found in $\widetilde{O}(n^7)$ time.
On the other hand, some special cases like graph density decomposition
can be solved in near-linear time via the minimum quadratic cost flow algorithm~\cite{DBLP:conf/focs/ChenKLPGS22}.
However, for large instances, it is more practical to use iterative approaches~\cite{DBLP:conf/www/DanischCS17,DBLP:conf/nips/HarbQC22}
to obtain approximation solutions.  

\noindent \emph{First-Order Iterative Methods for Convex Program.}
In view of Theorem~\ref{th:main_maximin}, we can 
use well-known first-order iterative methods to 
obtain approximate solutions for the convex program.
In the literature, there are two broad approaches.

\begin{compactitem}

\item \emph{Projected Gradient Descent Variants.}
For special cases like graph density decomposition,
when a gradient descent update step causes the solution
to go outside the feasible set, the projection operation
can be implemented very efficiently~\cite{DBLP:conf/nips/HarbQC22}.
Hence, for these special cases, gradient descent with momentum
can be used to get a fast approximation algorithm.

However, in the dual-modular setting, it is not clear if there is
a fast algorithm to project 
a tentative (infeasible) point in $\R^V \times \R^V$ back into the feasible 
set~$\Bfg$.

\item \emph{Frank-Wolfe Variants.} The idea of Frank-Wolfe is that
given a current solution $z^{(t)}$ for the objective function $\Phi$\
and feasible $\mcal{C}$,
a \emph{gradient oracle} is consulted
to get some feasible $s^{(t)} \in \arg \min_{s \in \mcal{C}} \langle s, \nabla \Phi(z^{(t)}) \rangle$.

Then, the updated solution $z^{(t+1)}$ is obtained as some convex combination
of $z^{(t)}$ and $s^{(t)}$.  The advantage is that no projection step is needed,
and hence is suitable for the dual-modular convex program
in Theorem~\ref{th:main_maximin}.

\end{compactitem}

\noindent \emph{Universal Gradient Oracle.} In Theorem~\ref{th:main_maximin}, we see that the convex program can be formulated with many possible choices
of convex~$\vartheta$. It turns out that the Frank-Wolfe approach in~\cite{DBLP:conf/esa/HarbQC23} used for linear $\g$ can be generalized to the dual-modular setting for any choice of~$\vartheta$.  Specifically,
given any current allocation $(x^{(t)}, y^{(t)})$, no matter what convex~$\vartheta$ is used,
it suffices to sort the elements~$v \in V$ in non-increasing
order of the induced densities $\frac{x^{(t)}(v)}{y^{(t)}(v)}$
to produce a permutation~$\sigma$.  The gradient oracle then returns~$(\f^\sigma, \g^\sigma)$, which forms a convex combination with
$(x^{(t)}, y^{(t)})$ to get an updated $(x^{(t+1)}, y^{(t+1)})$.

\noindent \emph{Approximation Guarantees.}
Even though different choices of the convex~$\vartheta$
in the objection function $\Phi_\vartheta$ would lead
to the same gradient oracle in Frank-Wolfe,
the theoretical approximation bounds do depend on the choice
of $\vartheta$. Given an estimation~$\widehat{\rho}$
for the target density vector~$\rho^*$,
we could consider the \emph{absolute error}~$\|\widehat{\rho} - \rho^*\|_2$,
and would depend on the normalization $\f(V) = \g(V) = 1$.

On the other hand, we can also consider
the coordinate-wise \emph{multiplicative error}~$\max_{u \in V} \epsilon_u$,
where $\epsilon_u := \frac{|\widehat{\rho}(u) - \rho^*(u)|}{\rho^*(u)}$.
Even though conventionally  Frank-Wolfe is analyzed under
absolute error,  it is easy to convert it to multiplicative error
if we have a lower bound on the correct densities as follows:
$\|\widehat{\rho} - \rho^*\|_2 \geq |\widehat{\rho}(u) - \rho^*(u)| \geq \epsilon_u \cdot \f_{\min}$,
where $\f_{\min} := \min_{u, \sigma} \frac{\f^\sigma(u)}{\f(V)}$.

We have analyzed the error using different choices of~$\vartheta$.  Here is a typical result using~$\vartheta(t) := t^2$.

\begin{theorem}[Approximation Guarantees of Frank-Wolfe]
\label{th:main_approx}
After $T$ steps of Frank-Wolfe,
the solution $(x^{(T)}, y^{(T)}) \in \mcal{B}^{\geq}_\f\times\mcal{B}^{\leq}_\g$
induces a density vector $\rho^{(T)} \in \R^V$ with the following approximation 
guarantees:

\begin{compactitem}

\item \emph{Absolute error}.

$$\| \rho^{(T)} - \rho^* \|_2 \leq \frac{O(1)}{\g_{\min}^{2.5} \cdot \sqrt{T+2}}.$$

\item \emph{Coordinate-wise Multiplicative error.}
For all $u \in V$,

$$\left|\frac{\rho^{(T)}_u - \rho^*_u}{\rho^*_u} \right| \leq \frac{O(1)}{\f_{\min} \cdot \g_{\min}^{2.5} \cdot \sqrt{T+2}}.$$

\end{compactitem}
\end{theorem}

\begin{remark}
For typical values $\f_{\min} = \Theta(\frac{1}{n})$
and $\g_{\min} = \Theta(\frac{1}{n})$,
it turns out that for several choices of~$\vartheta$,
we will get the same theoretical bounds:
achieving an absolute error takes $O(n^5)$ iterations,
while a multiplicative error needs $O(n^7)$ iterations.
However, each iteration will take $O(n \log n)$ time, due to the sorting step involved.
While the theoretical running times seem worse than
that for exact decomposition, experiments show that the empirical performance
of the iterative methods
are usually better in practice~\cite{DBLP:conf/esa/HarbQC23}.
\end{remark}

\ignore{
\noindent \textbf{Significance of the Frank-Wolfe Iterative Approach.}
Using the iterative approach, we can estimate every coordinate
of the density vector~$\rho^*$ with multiplicative error, which
can be used to give approximation solutions for the
dual-modular combinatorial contracts problem.  The details are given
in Section~\ref{sec:contracts_approx}.  
}

\noindent \textbf{Bonus Insights: Symmetry Between Reward and Cost.}
Besides giving approximate solutions,
another surprising aspect of the Frank-Wolfe approach is that
it reveals the symmetry between the two dual-modular functions
in an instance $(V; \f, \g)$ with strictly monotone $\f$ and $\g$.
We can reverse the roles of the two set functions
by considering supermodular $\overline{\g}(S) := \g(V) - \g(V \setminus S)$
and submodular $\overline{\f}(S) := \f(V) - \f(V \setminus S)$
to form its \emph{complementary} instance $(V; \overline{\g}, \overline{\f})$.
Such a symmetry between vertices and edges has been explored for the special case of graph density decomposition~\cite{DBLP:conf/innovations/ChanX25}, and a natural question is whether the two complementary dual-modular instances
will induce the same partition on~$V$, albeit in reversed order and with reciprocal densities.

The Frank-Wolfe gradient oracle actually gives a simple affirmative answer.
Observe that for any permutation~$\sigma$, if we denote $\overline{\sigma}$ as
its reversed order, then we have the identities:

 $\f^\sigma = \overline{\f}^{\overline \sigma} \in \mcal{B}^{\geq}_\f = \mcal{B}^{\leq}_{\overline \f}$
and $\g^\sigma = \overline{\g}^{\overline \sigma} \in
\mcal{B}^{\leq}_\g = \mcal{B}^{\geq}_{\overline \g}$.

Hence, given a current solution $(x, y)$, if we produce the order~$\sigma$ by sorting $v \in V$
in non-increasing order of $\frac{x_v}{y_v}$ in the instance $(V; \f, \g)$, it is equivalent
to the order $\overline{\sigma}$ by sorting in non-increasing order of $\frac{y_v}{x_v}$
in the instance $(V; \overline{\g}, \overline{\f})$.
Therefore, running Frank-Wolfe on the two complementary instances will produce exactly the same behavior,
with the roles of $x$ and $y$ exchanged.  Because the Frank-Wolfe procedure
is supposed to converge to the correct density vector, this simple observation\footnote{As promised earlier,
this implies that lexicographically optimality can equivalently be defined by the procedure that starts with
maximizing the minimum density among all agents.} proves the following statement.

\begin{corollary}[Symmetry Between Reward and Cost]
\label{cor:symmetric}
Given a dual-modular instance $(V; \f, \g)$ with strictly monotone~$\f$ and $\g$,
suppose $(V; \overline{\g}, \overline{\f})$ is its complementary instance
as defined above.  Then, the two instances produce density vectors $\rho^*$ and
$\overline{\rho}^*$ such that for all $u \in V$, 

$\rho^*(u) \cdot \overline{\rho}^*(u) = 1$.
\end{corollary}

\subsection{Related Work}

As aforementioned, the closest related works to 
our dual-modular setting $(V; \f, \g)$
considered the special case of linear~$\g$.
Fujishige~\cite{fujishige1980lexicographically}
considered the lexicographically optimal base
of a polymatroid by considering a convex program
using quadratic functions.  Jain and Vazirani~\cite{DBLP:journals/geb/JainV10}
studied the same setting from the perspective of submodularity utility allocation (SUA)
markets, where the linear $\g$ corresponds to a fixed budget per agent;
moreover, market equilibrium is defined by a unique optimal solution
to a convex program involving logarithmic functions.
As pointed out by Nagano~\cite{DBLP:conf/ipco/Nagano07}, those two problems are equivalent,
and details about the connection with the decomposition are explained
in the survey~\cite{Panditi2016}.  Harb et al.~\cite{DBLP:conf/esa/HarbQC23}
applied the Frank-Wolfe method to the convex program and readily derived an iterative algorithm
that converges to the optimal solution.

For the special case of (hyper)graphs,
the connection between the density decomposition, Fisher market equilibrium
and locally maximin allocation has been recently explored by
Chan and Xue~\cite{DBLP:conf/innovations/ChanX25}, who also showed
that a locally maximin allocation corresponds to a universal
minimizer over the class of divergences satisfying the data processing inequality.

The (generalized) densest subset problem~\cite{DBLP:journals/csur/LancianoMFB24} can be viewed as 
returning the subset with the largest density in the decomposition.
Kawase and Miyauchi~\cite{DBLP:conf/isaac/KawaseM16a} considered several density notions
on edge-weighted graphs, and the densest subset problem
for one of them can be solved in polynomial time and
falls under our dual-modular setting $(V; \f, \g)$ with $\g(S) := \phi(|S|)$ for some 
concave function~$\phi$. The function~$\g$ with same
form has also been considered~\cite{DBLP:conf/soda/ChekuriQT22},
but their focus is to show that \emph{iterative peeling} procedure can approximate
the densest subset with $O(1)$ ratio.

Interestingly, locally maximin solutions under the dual-modular framework
have implicitly appeared as technical tools in many works such as tree packing~\cite{DBLP:conf/stoc/Thorup08},
online matroid intersection maintenance~\cite{DBLP:conf/soda/BuchbinderGHKS24},
online submodular assignment problems~\cite{DBLP:conf/focs/HathcockJPSZ24},
quotient sparsification for submodular functions~\cite{DBLP:conf/soda/Quanrud24a}.
In these works, the technical parts would typically consider two set functions on
some ground set~$V$, where one function $\mathfrak{l}$ is linear
and the other may be a monotone super/sub-modular function~$\mathfrak{h}$.
However, more often than not, each of these works would have to re-define the
density decomposition or argue that the desired vector
is the optimal solution of some convex program (e.g, the one for the SUA market~\cite{DBLP:journals/geb/JainV10}).

Besides showing that the results can be extended to the more general
dual-modular setting, our work offers a comprehensive view when
one makes sense of how these tools are employed.
For instance, when $\mathfrak{l}$ is linear and $\mathfrak{h}$ is submodular,
we know from Corollary~\ref{cor:symmetric} that we can 
formulate this situation as either $(V; \mathfrak{l}, \mathfrak{h})$
or its complement instance $(V; \overline{\mathfrak{h}}, \mathfrak{l})$,
both of which have appeared in the literature.

\subsection{Paper Organization}

The introduction has presented all the main ideas and results of this paper.
Section~\ref{sec:prelim} gives more detailed descriptions of the background materials.  The main technical proofs are given in Section~\ref{sec:technical},
which shows the connection between locally maximin solutions and the convex program under the dual-modular framework.  In Section~\ref{sec:contracts},
we show how this framework can be used to analyze the dual-modular combinatorial contracts problem.  Section~\ref{sec:FW} gives the detailed analysis of the Frank-Wolfe iterative approach to approximate the convex program and the density vector.  Finally, in Section~\ref{sec:general_div},
we show that any locally maximin allocation is a universal minimizer under
the wide class of divergences satisfying the data processing inequality.

\section{Preliminaries}
\label{sec:prelim}

\begin{definition}[Dual-Modular Input Instance] 
\label{defn:input}
An instance $(V; \f, \g)$
consists of a finite \emph{ground} set $V$
and a pair of non-negative
monotone\footnote{A set function $h: 2^V \to \R$
is monotone if $A \subseteq B \subseteq V$
implies that $h(A) \leq h(B)$.} 
set functions $\f, \g: 2^V \to \R_{\geq 0}$
satisfying  $\f(\emptyset) = \g(\emptyset) = 0$.
In addition, we have:

\begin{compactitem}

\item $\f$ is supermodular, i.e.,
for all $A,B\subseteq V$,
 $\f(A)+ \f(B)\leq \f(A\cap B)+ \f(A\cup B)$.

\item $\g$ is submodular, i.e.,
for all $A,B\subseteq V$,
 $\g(A)+ \g(B)\geq \g(A\cap B)+ \g(A\cup B)$.
\end{compactitem}

\end{definition}

However, to avoid pathological situations,
we enforce the following technical assumption.

\begin{assumption}[Technical Assumption on $\g$]
\label{assume:g}
We assume that the submodular function $\g: 2^V \to \R_{\geq 0}$
is strictly monotone, i.e., for $A \subsetneq B$,
$\g(A) < \g(B)$.
\end{assumption}

\noindent \emph{Remark.}
  Strict monotonicity can be achieved by
picking arbitrarily small $\eta > 0$
and consider the perturbation
$\widetilde{\g}(S) =  \g(S) + \eta \cdot |S|$.
With this assumption,
it is guaranteed that for any $y \in \mcal{B}^{\leq}_\g$,
all coordinates of $y$ are non-zero.

\noindent \textbf{Normalization Convention.}
For notational convenience, we normalize the
functions with a multiplicative factor, and assume $\f(V) = \g(V) = 1$,

\noindent \textbf{Marginal Contribution.}
Given a set function $h: 2^V \to \R$,
the marginal contribution of a subset $S \subseteq V$ with respect
to $A \subseteq V$ is
$h(S | A) := h(S \cup A) - h(A)$.

\ignore{
\noindent \textbf{Non-Degenerate Elements.}
To avoid degenerate elements, we further assume that
for 
each $v \in V$,
both $\f(\{v\} |V - v ) > 0$ and $\g(\{v\}) > 0$ hold.
}


\noindent \emph{Equivalent Perspectives.}
In the literature, it is known
that for the special case when $\g$ is linear, 
the input instance in Definition~\ref{defn:input}
can be interpreted as inputs to
two closely related problems~\cite{fujishige1980lexicographically}: (i) density decomposition,
and (ii) convex program optimization.
It is straightforward to generalize the setting to submodular $\g$.

\subsection{Dual-Modular Density Decomposition}

We formalize the description of the density decomposition described
in Section~\ref{sec:intro}.
\ignore{
For the special case of linear $\g$,
Fujishige~\cite{fujishige1980lexicographically}
first considered the density decomposition
in the context
of lexicographically optimal base of polymatroids,
but this decomposition has been independently rediscovered many times
in the computer science community.
}
Recall that for $S \subseteq V$,
its density in the instance (as in Definition~\ref{defn:input})
is $\rho(S) := \frac{\f(S)}{\g(S)}$.
\ignore{
Chekuri et al.~\cite{DBLP:conf/nips/HarbQC22}
called the \emph{generalized densest supermodular subset problem}
as finding $S \subseteq V$ to maximize $\rho(S)$.
Just as the well-known case for linear $\g$,
}
The following fact implies that the maximal (with respect to 
set inclusion) densest subset is unique.

\begin{fact}
Suppose both $A$ and $B$ are densest subsets
in the instance $(V; \f, \g)$.
Then, both $A \cup B$ and $A \cap B$ (if $A \cap B \neq \emptyset$)
are densest subsets.
\end{fact}


\ignore{
\begin{fact}[Appendix \ref{unimax}]\label{fact1}
The maximal densest subset is unique and contains all densest subsets.
\end{fact}
}

\ignore{
So far in the literature, the density decomposition has been considered
when at least one of $\f$ and $\g$ is linear.
However, it is straightforward
to extend the description for a general instance
$(V; \f, \g)$, which can be formulated
by an iterative procedure.  First, find the 
maximal densest subset $S_1$, and
perform recursion on the subinstance
$(V \setminus S_1; \f(\cdot | S_1), \g( \cdot | S_1))$
until an empty ground set is reached.
In addition to a partition on $V$,
a corresponding \emph{density} vector $\mathfrak{\rho} \in \R^V$ is 
produced.  The formal description is given as follows.
}

\begin{definition}[Dual-Modular Density Decomposition and Density Vector]
\label{defn:density_decomp}
Given an instance $(V; \f, \g)$, the density decomposition
consists of a partition $V = \cup_i S_i$
and a density vector $\rho^* \in \R^V$
generated by the following iterative procedure.

\begin{compactitem}
\item For $i=0$, we use the convention $S_0:=\emptyset$.

\item For $i\geq 1$,
denote $S_{<i} := \cup_{j=0}^{i-1} S_j$,
and
consider the subinstance with ground set $V_i := V \setminus S_{<i}$ and the marginal set functions
$\f(\cdot | S_{<i})$ and $\g( \cdot | S_{<i})$.
Denote $\rho(S | S_{<i}) := 
\frac{\f(S|S_{<i})}{\g(S|S_{<i} )}$.

Find the corresponding maximal densest subset

$$S_i := \arg \max_{S\subseteq V_i} \rho(S | S_{<i}) .$$

For each $u \in S_i$, set $\rho^*(u) := \rho(S_i | S_{<i})$.

Stop if $V_i=S_i$; otherwise,
consider $i \gets i + 1$.
\end{compactitem}

Suppose $k$ is the total number of 
subsets in the decomposition $V = \cup_{i=1}^k S_i$.
For each $i$, we also denote
$\rho_i = \rho(S_i | S_{<i})$.
From Remark~\ref{remark:strict_decrease},
we see later that $\rho_1 > \rho_2 > \cdots > \rho_k$.

\noindent \emph{Density Decomposition Problem.}
Given the input instance, the problem is to compute or approximate
the above density vector $\rho^* \in \R^V$.

\end{definition}

\begin{remark}
Because of Assumption~\ref{assume:g}
for all $u \in V$, $0 \leq \rho^*_u < + \infty$.
\end{remark}


\subsection{Dual-Modular Market Allocation and Convex Program}

\noindent \textbf{Feasible Allocation.}
Given a dual-modular instance $(V; \f, \g)$,
we denote $\Delta_\f(V) := \{x \in \R^V_{\geq 0}: x(V) = \f(V)\}$,
where $x(S) := \sum_{u \in V} x_u$,
and define $\Delta_\g(V)$ in a similar way.
\ignore{
Since we consider normalized $\f(V) = \g(V) = 1$,
we may drop the subscript and consider $\Delta(V)$ as the collection
of distributions on $V$.}
A (feasible) \emph{allocation} is a pair $(x, y) \in
\mcal{B}^{\geq}_\f  \times \mcal{B}^{\leq}_\g$ of vectors
in the bases defined as follows:

\begin{compactitem}

\item \emph{Base Contrapolymatroid.}
$\mcal{B}^{\geq}_\f := \{x \in \Delta_\f(V) :
\forall S \subseteq V, x(S) \geq \f(S)\}$

\item \emph{Base Polymatroid.}
$\mcal{B}^{\leq}_\g := \{y \in \Delta_\g(V) :
\forall S \subseteq V, y(S) \leq \g(S)\}$

\end{compactitem}

\noindent \textbf{Fairness Notions.}
As mentioned in the introduction,
market equilibrium can be characterized
by the fairness notions \emph{lexicographically optimal} 
or \emph{locally maximin}.  The following simple argument
shows that the two notions are essentially equivalent,
once we have established the existence of a locally maximin
allocation.  The proof does not need to use dual-modularity
and just needs a monotone $\f$ and a strictly monotone $\g$,
which ensures that every $y \in \mcal{B}^{\leq}_\g$
has positive coordinates.

\begin{fact}[Equivalence of Fairness Notions]
\label{fact:equiv_fair}
Suppose a locally maximin allocation exists.
Then, an allocation is lexicographically optimal \emph{iff}
it is locally maximin; moreover, any such allocation
will induce the same density vector.
\end{fact}

\begin{proof}
Suppose a locally maximin allocation~$(x^*, y^*)$ exists.  
We prove 
by induction on the number~$k$ of distinct induced
densities in $\{\frac{x^*_u}{y^*_u}: u \in V\}$.

The base case $k=1$ is trivial, because
both reward and cost are constant-sum.
Hence, the allocation~$(x^*, y^*)$ in which every agent has equal induced density
must be lexicographically optimal.
Conversely, 
any lexicographically optimal allocation
must induce 
the same multi-set of densities as
$(x^*, y^*)$, which we have just proved to be lexicographically optimal;
hence, it induces
the same density for every agent, which satisfies
the local maximin condition.

For the inductive step $k \geq 2$,
consider $S_1 := \arg \max_{u \in V}  \frac{x^*_u}{y^*_u}$.
Because $(x^*, y^*)$ is locally maximin,
we have $x^*(S_1) = \f(S_1)$ and $y^*(S_1) = \g(S_1)$.
Moreover, in any feasible allocation,
there exists an agent in $S_1$ whose induced density must be
at least $\rho_1 := \frac{\f(S_1)}{\g(S_1)}$,
and the only case where every agent in $S_1$ has induced density not exceeding
$\rho_1$ is that the agents in the group $S_1$
are in the worst situation in terms of reward and cost, and every one in $S_1$ has induced density exactly~$\rho_1$.

Since we have a candidate allocation $(x^*, y^*)$,
it follows that any lexicographically optimal allocation
$(\widehat{x}, \widehat{y})$ must 
allocate $\widehat{x}(S_1) = \f(S_1)$ 
and $\widehat{y}(S_1) = \g(S_1)$,
and induce the density $\rho_1$ for every agent in $S_1$.

We next apply the induction hypothesis
to the residue instance $(V \setminus S_1; \f( \cdot | S_1),
\g(\cdot |S_1)$ and the restricted allocation $(x^*|_{S_1},
y^*|_{S_1})$.
Hence, we conclude that any lexicographically
optimal allocation must also be locally maximin and
induce the same density vector
as $(x^*, y^*)$, which finishes the inductive step.
\end{proof}

\noindent \textbf{Expressing a Solution with a Pair of Permutation Distributions.}
As mentioned in the introduction,
we consider the collection $\mcal{S}_V$ of permutations on $V$.
Given $\sigma \in \mcal{S}_V$,
we use $u \prec_\sigma v$ to mean $u$ appears earlier than $v$
in the permutation~$\sigma$.
Given a set function $h: 2^V \to \R$,
we denote $h^\sigma \in \R^V$ as the vector
defined by: $h^\sigma(u) := h(\{u\} | \{w\in V: w\prec_{\sigma} u\}), \forall u \in V.$

We use $\Delta(\mcal{S}_V) := \{ \vecp = (p_\sigma \geq 0)_{\sigma \in \mcal{S}_V}:
\sum_{\sigma \in \mcal{S}_V} p_\sigma = 1  \}$ to 
denote the collection of permutation distributions.
We naturally extend the notation
to $h^{\vecp} = \sum_{\sigma \in \mcal{S}_V} p_\sigma \cdot h^\sigma \in \R^V$.

The following fact implies that we can use a pair $(\vecp, \vecq)$ of permutation
distributions to represent an allocation.  Henceforth,
we use the terminology \emph{solution} for $(\vecp, \vecq)$.

\begin{fact}[Permutation Distributions as Solution~\cite{SFMCP}]
\label{fact:perm_polytope2}
For all $(\vecp, \vecq) \in \Delta(\mcal{S}_V) \times \Delta(\mcal{S}_V)$,
$(\f^\vecp, \g^\vecq) \in \mcal{B}^{\geq}_\f \times \mcal{B}^{\leq}_\g$
is a feasible allocation.
Conversely, for any $(x, y) \in \mcal{B}^{\geq}_\f \times \mcal{B}^{\leq}_\g$, there exists a pair $(\vecp, \vecq) \in \Delta(\mcal{S}_V) \times \Delta(\mcal{S}_V)$ of permutation distributions such that $(x, y) = (\f^\vecp, \g^\vecq)$.
\end{fact}

\noindent \emph{Locally Maximin Solution.}
Observe that we could have defined a locally maximin 
solution $(\vecp, \vecq)$ naturally as one such that the induced
$(\f^\vecp, \g^\vecq)$ is locally maximin.
However, to facilitate our proofs, we will use a stronger
locally maximin condition for a solution in 
$\Delta(\mcal{S}_V) \times \Delta(\mcal{S}_V)$.
For $v \in V$,
we denote its density induced by $(\vecp , \vecq)$
as $\rho^{(\vecp, \vecq)}(v) := \frac{\f^\vecp(v)}{\g^\vecq(v)}$.

\begin{definition}[Locally Maximin Solution]
\label{defn:local_maximin}
Given an instance $(V; \f, \g)$,
a solution $(\vecp, \vecq) \in \Delta(\mcal{S}_V)^2$
is locally maximin if for all $u, v \in V$
such that $\rho^{(\vecp, \vecq)}(u) > \rho^{(\vecp, \vecq)}(v)$
and 
$\sigma \in \mcal{S}_V$ such that $p_\sigma > 0$ or $q_\sigma > 0$,
it holds that $u \prec_\sigma v$.
\end{definition}

\begin{remark}
For supermodular $\f$ and submodular $\g$,
if $u \notin B$ and $A \subseteq B$,
then we have $\f(\{u\}|A) \leq \f(\{u\}|B)$
and $\g(\{u\}|A) \geq \g(\{u\}|B)$.
Therefore, an agent will intuitively be more satisfied
if it appears later in a permutation.

The high-level rationale for the local maximin condition
is that if $v$ is strictly less satisfied than $u$
in $(\vecp, \vecq)$, then $v$ is already treated with higher priority
than $u$ in the sense that
 $v$ must appear after $u$
in every permutation that has non-zero weight in $(\vecp, \vecq)$.
\end{remark}

\begin{fact}[Locally Maximin Condition for Two Representations]
\label{fact:local_represent}
Every locally maximin solution
$(\vecp, \vecq)$ induces a locally maximin allocation~$(\f^\vecp, \g^\vecq)$.
Conversely, for every locally maximin allocation~$(x, y)$,
there exists a locally maximin solution
$(\vecp, \vecq)$ such that $(x, y) = (\f^\vecp, \g^\vecq)$.

(However, it is possible that there is some solution~$(\vecp, \vecq)$
that is not locally maximin, but the induced allocation~$(\f^\vecp, \g^\vecq)$
is locally maximin.)
\end{fact}

\begin{proof}
We outline an argument for this simple fact.
For the first direction,
observe that given a partition $V = \cup_{i=1}^k S_i$,
for any permutation $\sigma \in \mcal{S}_V$
such that the parts obey $S_1 \prec_\sigma S_2 \prec_\sigma \cdots
\prec_\sigma S_k$ (but elements within $S_i$ can 
be ordered arbitrarily),
we have for each $i$,  $\sum_{u \in S_i} \f^\sigma(u) = \f(S_i | S_{<i})$
and $\sum_{u \in S_i} \g^\sigma(u) = \g(S_i | S_{<i})$.

For the second direction, given a locally maximin $(x, y)$,
we partition $V = \cup_{i=1}^k V_i$ such that
elements in each part $V_i$ have the same induced density,
where the indices are sorted according to the densities
such that $V_1$ has the largest density.
Because $(x, y)$ is locally maximin, we can decompose $(V; \f, \g)$ into $k$ subinstances,
where the $i$-th instance
corresponds to $(V_i; \f(\cdot | V_{<i}), \g(\cdot | V_{<i}))$
and the restricted $(x|_{V_i}, y|_{V_i})$ is a feasible allocation.
Fact~\ref{fact:perm_polytope2} gives a corresponding pair $(\vecp_i, \vecq_i) \in \Delta(\mcal{S}_{V_i})^2$
of distributions that will induce $(x|_{V_i}, y|_{V_i})$.
Taking the product distribution $(\otimes_{i=1}^k \vecp_i, \otimes_{i=1}^k \vecq_i)$, which is a locally maximin solution, will induce $(x, y)$, as required.
\end{proof}

\noindent \textbf{Convex Program Formulation.}
As mentioned in the introduction,
it is straightforward to generalize the convex program
formulation~\cite{DBLP:conf/ipco/Nagano07,Panditi2016} for linear~$\g$
to the dual-modular setting.  We recall the notion of convex divergence.

\begin{definition}[$\vartheta$-Divergence]
\label{defn:div}
Given a convex function~$\vartheta$,
the $\vartheta$-divergence between two
distributions $x, y$ in $\Delta(V)$ is defined as:

$$\D_\vartheta(x \| y) := \sum_{u \in V} y_u \cdot \vartheta(\frac{x_u}{y_u}),$$

\end{definition}

\noindent \textbf{Technical Remarks.}
Typically, a divergence is used to compare two distributions;
hence, we can consider the normalization
$\f(V) = \g(V) = 1$.  Furthermore,
in the literature, a constant is added to convex function
to ensure $\vartheta(1) = 0$, but it is not necessary in our setting.
When $y$ has zero coordinates,
we use the convention that (i) $0 \cdot \vartheta(\frac{0}{0}) = 0$, 
and (ii) for $a > 0$, $0 \cdot \vartheta(\frac{a}{0}) = a \lim_{x \to + \infty} \frac{\vartheta(x)}{x}$.  But, due to Assumption~\ref{assume:g},
we shall only consider $y$ with non-zero coordinates.


\begin{definition}[Dual-Modular Convex Program]
\label{defn:general_convex}
Given a dual-modular instance $(V; \f,\g)$ as in Definition~\ref{defn:input} and 
a convex function $\vartheta: \mathbb{R}_{\geq 0}\rightarrow \mathbb{R}$,
define the following convex program:

\begin{align*}
\min\quad \D_\vartheta(x \| y)  &\\
(x, y) \in \mcal{B}^{\geq}_\f \times \mcal{B}^{\leq}_\g &
\end{align*}

In view of Fact~\ref{fact:perm_polytope2},
we will have an equivalent formulation in terms of permutation distributions:

\begin{align*}
 \mathsf{CP:}\min\quad \Phi_\vartheta(\vecp, \vecq) &:= \D_\vartheta(\f^{\vecp} \| \g^{\vecq})  \\
(\vecp, \vecq) &\in \Delta(\mathcal{S}_V)^2 
\end{align*}

We also denote the induced density vector $\rho^{(\vecp, \vecq)} \in \R^V$,
where $\rho^{(\vecp, \vecq)}(u) = \frac{\f^{\vecp}_u}{\g^{\vecq}_u}$,
for $u \in V$.

When the context is clear, we may drop the subscript
and write the objective function as $\Phi$.
Moreover, we may overload the notation
and also write $\Phi(x, y) = \D_\vartheta(x \| y)$.

\ignore{
We also denote the induced density vector $\rho^{(\vecp, \vecq)} \in \R_*^V$,
where $\R_* := \R \cup \{\bot, +\infty\}$,
as follows:

for $u \in V$,
 $\rho^{(\vecp, \vecq)}(u) =  
\begin{cases}
\bot, & \text{if } \f^{\vecp}_u = \g^{\vecq}_u = 0; \\
\frac{\f^{\vecp}_u}{\g^{\vecq}_u}, & \text{otherwise.}
\end{cases}
$

We denote $V^{(\vecp, \vecq)} := \{u \in V: \rho^{(\vecp, \vecq)}(u) \neq \bot\}$.
}
\end{definition}

\begin{remark}
The objective function in Definition~\ref{defn:general_convex}
is convex because for any $c, d \in \R^m$,
the function $(x, y) \in \R^m \times \R^m \mapsto \langle d, y \rangle \cdot \vartheta(\frac{\langle c, x \rangle}{\langle d, y \rangle })$
is convex due
to property of \emph{perspective functions}~\cite{CVOP}.
\end{remark}

\noindent \textbf{Extra Equality Constraint.}
We later show that the minimum of the convex program $\mathsf{CP}$ in 
Definition~\ref{defn:general_convex}
can be attained with an extra equality constraint $\vecp = \vecq$.
Hence, when there is no ambiguity, we will overload
the notation $\Phi(\vecp) :=  \Phi(\vecp, \vecp)$
and $\rho^{\vecp} := \rho^{(\vecp, \vecp)}$.

\subsection{Iterative Approximation Methods}

\ignore{
\noindent \emph{Exact Decomposition vs Approximation.}
As mentioned in~~\cite{DBLP:conf/nips/HarbQC22},
the generalized densest subset problem can be solved in polynomial time
by submodular function optimization.
Therefore, the exact density decomposition in Definition~\ref{defn:density_decomp}
can also be found in polynomial time, while some special cases
can be solved in near-linear time via maximum flow~\cite{DBLP:conf/www/TattiG15}.
However, for large instances, it is more practical to use iterative approaches~\cite{DBLP:conf/nips/HarbQC22}
to obtain an approximation.  
}

\ignore{
\noindent \textbf{Abstract Iterative Procedure.}
In view of Fact~\ref{fact:convex_density},
first-order iterative methods for convex programs
have inspired several
iterative procedures
to approximate the density vector in Definition~\ref{defn:density_decomp}
which are instantiations of the following abstract procedure.
(For notation simplicity, we focus on the special case $\g(S) = |S|$.)
}

We first outline the Frank-Wolfe variants described in~\cite{DBLP:conf/esa/HarbQC23}
that assume a linear~$\g$.
For notation simplicity, we illustrate the methods under the special case $\g(S) = |S|$.

\begin{algorithm}[H]
\caption{Abstract Frank-Wolfe Procedure}\label{alg:abstract}
\begin{algorithmic}[1]
\State \textbf{Input:} Supermodular $\f : 2^V \to \R$;
step size $\gamma: \Z_{\geq 0} \to [0,1]$;
number $T$ of iterations.

\State    \text{\textbf{Initialize}: } Pick arbitrary $x^{(0)} \in \mcal{B}^{\geq}_\f $.
 
\For{ $k \gets 0 \text{ to } T-1$}

  
		\State $d^{(k+1)}\gets \textsc{Abstract-Greedy}(\f, \gamma_k, x^{(k)}) \in \mcal{B}^{\geq}_\f $
     \State $x^{(k+1)}\gets (1 - \gamma_k) \cdot x^{(k)}+\gamma_k \cdot d^{(k+1)}$ \label{ln:update}
 
\EndFor
\State \text{\textbf{return: $x^{(T)}$} }  
\end{algorithmic}
\end{algorithm}

\begin{remark}
Because $\f$ is supermodular,
if $u \notin B$ and $A \subseteq B$,
then we have $\f(\{u\}|A) \leq \f(\{u\}|B)$.
Therefore, if we modify the permutation~$\sigma$ by moving some~$u$ earlier,
the marginal contribution $\f^\sigma(u)$ will either drop or stay the same.

The high-level rationale for 
\textsc{Abstract-Greedy}$(\f, \gamma_k, x^{(k)})$
is to pick $\f^\sigma$ from a permutation~$\sigma$ such that
elements~$u$ with higher value $x^{(k)}(u)$ will tend to appear earlier in~$\sigma$.
\end{remark}

To instantiate a specific iterative procedure from Algorithm~\ref{alg:abstract},
it suffices to specify (i)~the \textsc{Abstract-Greedy} procedure
and (ii)~the step size function~$\gamma$.

\ignore{
\noindent \textbf{Select from $\mcal{B}^{\geq}_\f $ using Permutation.}
For reasons that will be apparent later,
we will identify the vertices of $\mcal{B}^{\geq}_\f$
with permutations of $V$.
We use $\mcal{S}_V$ to denote the collection 
of permutations on $V$.  Given $\sigma \in \mcal{S}_V$,
we use $u \prec_\sigma v$ to mean $u$ appears earlier than $v$
in the permutation~$\sigma$.
Given a set function $h: 2^V \to \R$,
we denote $h^\sigma \in R^V$ as the vector
defined by: 

$$h^\sigma(u) := h(\{u\} | \{w\in V: w\prec_{\sigma} u\}), \forall u \in V.$$

\begin{fact}[~\cite{DBLP:conf/soda/ChekuriQT22,SFMCP}]
\label{fact:perm_polytope}
The vectors $\f^\sigma$ over $\sigma \in \mcal{S}_V$
are the vertices of $\mcal{B}^{\geq}_\f = \mathsf{conv}(\{\f^\sigma:   \sigma \in \mcal{S}_V\})$.
\end{fact}
}

\ignore{
In Algorithm~\ref{alg:abstract},
if each vector returned by \textsc{Abstract-Greedy}
has the form $\f^\sigma$ for some permutation~$\sigma \in \mcal{S}_V$,
then the iterative procedure actually implicitly maintains
a distribution of permutations.
With this insight and Fact~\ref{fact:perm_polytope}, we will describe an equivalent convex program from the perspective
of permutation distributions.

\noindent \textbf{Notation.}
We denote $\Delta(\mcal{S}_V) := \{ \vecp = (p_\sigma \geq 0)_{\sigma \in \mcal{S}_V}:
\sum_{\sigma \in \mcal{S}_V} p_\sigma = 1  \}$.
Given a set function $h: 2^V \to \R$,
we naturally extend the notation $h^\sigma \in \R^V$
to $h^{\vecp} = \sum_{\sigma \in \mcal{S}_V} p_\sigma \cdot h^\sigma \in \R^V$.
}

Because of Fact~\ref{fact:convex_program}, 
we will apply Frank-Wolfe to the following objective function:


$$\Phi(x) = \sum_{u \in V} x_u^2, x \in \mcal{B}^{\geq}_\f.$$

\noindent \textbf{Frank-Wolfe Variant.}
In Algorithm~\ref{alg:abstract},
one should choose
the step size function $\gamma_k := \frac{2}{k+2}$
and implement \textsc{Abstract-Greedy} with the following
\emph{gradient oracle}:

$$d^{(k+1)} \gets \arg \min_{x \in \mcal{B}^{\geq}_\f} \langle x, \nabla \Phi(x^{(k)}) \rangle
=  2 \langle x, x^{(k)} \rangle .$$

As mentioned in~\cite{DBLP:conf/esa/HarbQC23},
this can easily be achieved by sorting the elements in~$v \in V$ according to their values given
by $x^{(k)}_v$.

\begin{fact}
Given $x^{(k)} \in \mcal{B}^{\geq}_\f$,
suppose $\sigma \in \mcal{S}_V$ is a permutation obtained by
sorting the elements~$u$ in  $V$ in non-increasing order
of $x^{(k)}(u)$, where ties are resolved arbitrarily.
Then, $\f^\sigma \in \arg \min_{x \in \mcal{B}^{\geq}_\f}  \langle x, x^{(k)} \rangle$.
\end{fact}

\noindent \emph{Convergence Rate.}
In terms of the number~$T$ of iterations, it is known that Frank-Wolfe
can produce a convergence rate of $\Phi(\rho^{(T)}) - \Phi(\rho^*) \leq O(\frac{1}{T})$.

\noindent \textbf{\textsc{Greedy}$++$ Variant~\cite{DBLP:conf/esa/HarbQC23}.}
In view of line~\ref{ln:update} of Algorithm~\ref{alg:abstract},
one may want to choose a permutation based on the result of the update.
The \textsc{Greedy}$++$ variant chooses $\gamma_k := \frac{1}{k+1}$,
and the following way to pick a permutation $\sigma$ to return $\f^\sigma$ for the
\textsc{Abstract-Greedy} step.

\begin{algorithm}[H]
\caption{\textsc{PermGreedy}++}\label{alg:permgreedy}
\begin{algorithmic}[1]
\State \textbf{Input:} Supermodular $\f : 2^V \to \R$;
step size $\gamma_k \in [0,1]$;
current solution $x^{(k)} \in \mcal{B}^{\geq}_\f$.

\State \textbf{Initialize} $\sigma \gets$ empty list;   $W \gets V$
\While{ $W \neq \emptyset$}

 \State $u \gets \text{arg}\min\limits_{u \in W}
        \{  (1 - \gamma_k) \cdot x^{(k)}(u) + \gamma_k \cdot \f(\{ u \} | W - u)   \}$
     
		\State Insert $u$ at the beginning of the list: $\sigma \gets (u \oplus \sigma)$ 
		
    \State $W \gets W-u$
 
\EndWhile
\State \text{\textbf{return} permutation}  $\sigma \in \mcal{S}_V$
\end{algorithmic}
\end{algorithm}

\noindent \emph{Convergence Rate.}
While \textsc{Greedy}$++$ seems to perform better empirically than
the basic Frank-Wolfe variant,
currently only a slightly worse convergence rate in terms of $T$ can be proved:
$\Phi(\rho^{(T)}) - \Phi(\rho^*) \leq O(\frac{\log T}{T})$.

\ignore{
\begin{minipage}[t]{0.5\textwidth}
\begin{algorithm}[H]
\caption{\scriptsize {WEIGHTED-SUPERGREEDY ~\cite{DBLP:conf/esa/HarbQC23}}}\label{wspgalg}
\begin{algorithmic}[1]
\State \footnotesize  \text{{\footnotesize  Input}: } \text{\footnotesize  Supermodular $f:2^{V}\rightarrow \mathbb{R}_{\geq 0},w(u)$ for $u\in V$}
\State \footnotesize  $V' \gets V$
\State \footnotesize  \textbf{Initialize} $\hat{d}(u)=0$ for all $u\in V$.
\While{ $|V'|>1$}

 \State $u\gets \text{arg}\min\limits_{u\in V'}\{w(u)+{f(V')}-{f(V'-v)}\}$
     \State $\hat{d}(u)\gets {f(V')}-{f(V'-v)}$
    \State $V'\gets V'-u$
 
\EndWhile
\State \text{\textbf{return: $\hat{d}$} }  
\end{algorithmic}
\end{algorithm}
\end{minipage}
\hfill
\begin{minipage}[t]{0.5\textwidth}
\begin{algorithm}[H]
\caption{\scriptsize  {SUPERGREEDY++~\cite{DBLP:conf/esa/HarbQC23}}}\label{spgalg}
\begin{algorithmic}[1]
\State \footnotesize \text{{Input}: } \text{$(f,T)$}
\State \footnotesize   \text{\textbf{Initialize}: } $b^{(0)}\gets \text{\tiny {WEIGHTED-SUPERGREEDY}}(V,\textbf{0})$
 
\For{ $k\gets 0 \text{ to } T-1$}

 \State $\gamma \gets \frac{1}{k+1}$
     \State $d^{(k+1)}\gets \text{\tiny{WEIGHTED-SUPERGREEDY}}(V,(k+1)b^{(k)}) $
     \State $b^{(k+1)}\gets (1-\gamma)b^{(k)}+\gamma d^{(k+1)}$
 
\EndFor
\State \text{\textbf{return: $b^{(T)}$} }  
\end{algorithmic}
\end{algorithm}
\end{minipage}
}

\section{Analysis of Dual-Modular Convex Program}
\label{sec:technical}

We formalize the arguments outlined in the introduction.
The main technical results
are as follows. 

\begin{enumerate}
\item
From Lemma~\ref{lemma:maximin_opt}, we conclude that any optimal allocation $(x, y)$ to the
convex program in Definition~\ref{defn:general_convex} with 
strictly convex~$\vartheta$ must be locally maximin.

\item
From Lemma~\ref{lemma:maximin_density},
we conclude that any locally maximin allocation $(x, y)$ must induce the density
vector $\rho^*$ from the density decomposition, and all such locally maximin allocations have the same objective value in the convex program.
\end{enumerate}

Together with the equivalence between lexicographically optimality
and local maximin condition proved in Fact~\ref{fact:equiv_fair},
we have proved Theorem~\ref{th:main_market}.

\subsection{Structural Properties of Optimal Solutions}


\noindent \textbf{Notation.} Given
disjoint subsets $A$ and $B \subseteq V$
and a permutation $\sigma \in \mcal{S}_V$,
the notation $A \prec_\sigma B$ means
that for all $u \in A$ and $v \in B$,
$u \prec_\sigma v$.
Similarly, given $\vecp \in \Delta(\mcal{S}_V)$,
the notation 
$A \prec_{\vecp} B$ means that
$p_\sigma > 0$ implies that 
$A \prec_\sigma B$;
moreover, given $(\vecp, \vecq) \in \Delta(\mcal{S}_V)^2$,
the notation $A \prec_{(\vecp, \vecq)} B$
means that both $A \prec_{\vecp} B$ and 
$A \prec_{\vecq} B$ hold.

We remark that we need Assumption~\ref{assume:g}
to ensure that any solution $(\vecp, \vecq)$
can induce a density vector $\rho^{(\vecp, \vecq)} \in \R^V$,
because $\g^{\vecq}$ has non-zero coordinates.

\begin{lemma}[Existence of Locally Maximin Optimal Solutions]
\label{lemma:maximin_opt}
Suppose $(\vecp, \vecq) \in \Delta(\mcal{S}_V)^2$ is an optimal solution
to the convex program in Definition~\ref{defn:general_convex}.
Then, there exists a locally maximin optimal solution $(\widehat{\vecp}, \widehat{\vecq})$.

If $\vartheta$ is strictly convex,
then in $(\widehat{\vecp}, \widehat{\vecq})$,
the induced allocation and densities are the same as before,
i.e.,
$(\f^{\vecp}, \g^{\vecq}) =  (\f^{\widehat{\vecp}}, \g^{\widehat{\vecq}})$
and
  $\rho^{({\vecp}, {\vecq})} = \rho^{(\widehat{\vecp}, \widehat{\vecq})} \in \R^V$;
if we remove the strict convexity requirement for $\vartheta$ in the objective function,
then the resulting locally maximin optimal solution
may not preserve the induced densities.

Moreover, the same result holds for the convex program with the extra equality constraint $\vecp = \vecq$. For instance, if $\vartheta$ is strictly convex, we can also ensure the locally optimal solution $\widehat{\vecp} = \widehat{\vecq}$ preserves the allocation and induced densities.

%
%
\end{lemma}

\begin{proof}
We first prove the case for the convex program with the extra equality constraint $\vecp = \vecq$, whose proof is actually more general.
We use a standard exchange argument to gradually transform
the given optimal~$\vecp$ into an optimal solution with the desired properties.

Suppose there exists $p_\sigma > 0$
and adjacent $v \prec_\sigma u$ such that $\rho^{\vecp}_u >  \rho^{\vecp}_v$.
Let $\widehat{\sigma}$ be the permutation
obtained from $\sigma$ by swapping the positions of $v$ and $u$,
while the positions of all other elements~$w \notin \{v, u\}$ remain unchanged;
note that $\f^\sigma_w = \f^{\widehat{\sigma}}_w$
and $\g^\sigma_w = \g^{\widehat{\sigma}}_w$.

The idea is to consider, for small $\epsilon > 0$,
another solution $\widehat{\vecp}(\epsilon)$,
defined by:

 $q_\sigma(\epsilon) = p_\sigma - \epsilon$
and $q_{\widehat{\sigma}}(\epsilon) =  p_{\widehat{\sigma}} + \epsilon$,
where the coordinates for other permutations are unchanged.

Next, we simplify the notation.  
Denote $F_u := \f^{\vecp}_u$, $F_v := \f^{\vecp}_v$,
$G_u := \g^{\vecp}_u$ and $G_v := \g^{\vecp}_v$.

Moreover, $R_u := \frac{F_u}{G_u} > R_v := \frac{F_v}{G_v}$.

Define $F := \f^\sigma_u - \f^{\widehat\sigma}_u =
 \f^{\widehat\sigma}_v - \f^\sigma_v \geq 0$ and
$G := \g^{\widehat\sigma}_u - \g^\sigma_u = 
\g^\sigma_v - \g^{\widehat\sigma}_v \geq 0$,
where the equalities hold because only $u$ and $v$ are swapped
to change from $\sigma$ to $\widehat\sigma$,
and the inequalities hold due to the supermodularity of $\f$
and the submodularity of $\g$.

Observe that in the objective function $\Phi$ in Definition~\ref{defn:general_convex},
the contribution to the sum from vertices $w \notin \{u ,v\}$ does not change.
Hence, we will focus on the terms related to $u$ and $v$ only
in $\Phi(\widehat{\vecp}(\epsilon))$.

Using the simplified notation, 
observe that: 

$\f^{\widehat{\vecp}(\epsilon)}_u = F_u - \epsilon F$,
$\f^{\widehat{\vecp}(\epsilon)}_v = F_v + \epsilon F$,

$\g^{\widehat{\vecp}(\epsilon)}_u = G_u + \epsilon G$,
$\g^{\widehat{\vecp}(\epsilon)}_v = G_v - \epsilon G$.

Hence, we define the function, for sufficiently small $\epsilon \geq 0$,

$$\phi(\epsilon) := (G_v - \epsilon G) \cdot \vartheta(\frac{F_v + \epsilon F}{G_v - \epsilon G}) + (G_u + \epsilon G) \cdot \vartheta(\frac{F_u - \epsilon F}{G_u + \epsilon G}).$$

Our goal is to show that $F = G = 0$.  Suppose at least one of $F$ and $G$ is positive.  Then, for sufficiently small $\epsilon > 0$, we have the following
strict inequalities:

$$a := \frac{F_v}{G_v} < b := \frac{F_v + \epsilon F}{G_v - \epsilon G}
< c:= \frac{F_u - \epsilon F}{G_u + \epsilon G} < d := \frac{F_u}{G_u},$$

and  $$1 > \mu := \frac{G_v}{G_v + G_u} > \lambda := \frac{G_v - \epsilon G}{G_v + G_u} > 0.$$

Check that $\mu a + (1 - \mu) d = \lambda b + (1 - \lambda) c$.
Then, strict convexity of $\vartheta$ implies
that:

 $$\mu \vartheta(a) + (1 - \mu) \vartheta(d) > 
\lambda \vartheta(b) + (1 - \lambda) \vartheta(c),$$

which is equivalent to $\phi(0) > \phi(\epsilon)$.  But, this
contradicts $\phi(0) \leq \phi(\epsilon)$, because $\vecp$ is an optimal solution.

\ignore{
Note that because $\vecp$ is an optimal solution,
$\phi'(0) \geq 0$.  We shall show that equality actually holds.
Elementary calculus gives the following:

$\phi'(0) = G \cdot (\zeta(R_u) - \zeta(R_v)) - F \cdot(\vartheta'(R_u) - \vartheta'(R_v))$,

where $\zeta(t) := \vartheta(t) - t \cdot \vartheta'(t)$ is a strictly decreasing function.
This is because, for $0 \leq x < y$,
$\vartheta(y) - \vartheta(x) <  (y - x) \vartheta'(y)  \leq
y \vartheta'(y) - x \vartheta'(x)$.
Note that the first strict inequality follows from the mean value theorem and the strict convexity
of $\vartheta$.

Note that because $R_u > R_v \geq 0$,
we have $\zeta(R_u) - \zeta(R_v) < 0$ and $\vartheta'(R_u) - \vartheta'(R_v) > 0$.
Hence, $\phi'(0) \geq 0$ implies that $F = G = 0$, 
}

This  means that $F = G = 0$ and
$\phi(\epsilon) = G_u \cdot \vartheta(R_u) + G_v \cdot \vartheta(R_v)$
does not change as $\epsilon$ increases from 0.

From this we conclude that we can transfer all the weight from $\sigma$ to $\widehat\sigma$ in $\vecp$
to produce
the resulting solution $\widehat{\vecp}$ (with
 $\widehat{p}_{\widehat \sigma} := p_\sigma + p_{\widehat \sigma}$
and $\widehat{p}_{\sigma} := 0$)
that  is still optimal, and the induced allocation and densities for all elements do not change.

Because both $V$ and $\mcal{S}_V$ are finite,
applying this exchange argument repeatedly leads to the desired result.

\noindent \textbf{Strict Convexity Requirement.}
Observe that if we remove the strict convexity requirement for $\vartheta$ above,
then we would simply have $\phi(0) = \phi(\epsilon)$, for small enough $\epsilon > 0$.  This means as we increases $\epsilon$ from 0,
the objective function does not change, but the induced densities for $u$ and $v$ might change.

\noindent \textbf{Extension to General Convex Program without Equality Constraint
$\vecp = \vecq$.}  Note that we can modify $\vecp$ and $\vecq$ separately.
The proof is simpler, because in the first case, we set $G = 0$, where in the second case, we have $F=0$.

\end{proof}

\begin{lemma}[Locally Maximin Solution Gives Density Decomposition]
\label{lemma:maximin_density}
Given an instance $(V; \f, \g)$,
suppose $(\vecp, \vecq) \in \Delta(\mcal{S}_V)^2$
is a locally maximin solution.
Then, the induced density vector $\rho^{(\vecp, \vecq)}$
coincides 
with the density vector $\rho^*$ from the density decomposition $V = \cup_{i=1}^k S_i$
in Definition~\ref{defn:density_decomp}.

As an immediate corollary, all locally maximin solutions $(\vecp, \vecq)$
have the same objective value:

$\Phi(\vecp, \vecq) = \sum_{i=1}^k \g(S_i | S_{<i}) \cdot \vartheta(\frac{\f(S_i | S_{<i})}{\g(S_i | S_{<i})})$.
\end{lemma}

\begin{remark}[Strictly Decreasing Densities in Decomposition]
\label{remark:strict_decrease}
In this proof, we also get the result that 
the densities obtained in the decomposition
is strictly decreasing: $\rho_1 > \rho_2 > \cdots > \rho_k$,
where $\rho_i = \frac{\f(S_i | S_{<i})}{\g(S_i | S_{<i})}$.
\end{remark}

\begin{proof}
We prove by induction on 
the number $k$ of parts
in the density decomposition $V = \cup_{i=1}^k S_i$
in Definition~\ref{defn:density_decomp}.

\noindent \textbf{Base Case} $k=1$, i.e. $V = S_1$
is the maximal densest subset
with density $\rho_1 = \frac{\f(V)}{\g(V)}$.
Suppose, for contradiction's sake,
assume that a locally maximin solution $(\vecp, \vecq)$
induces a proper subset
$S^{*}:=\argmax_{u\in V} \rho^{(\vecp, \vecq)}(u) \subsetneq V$.

Observe that since $(\vecp, \vecq)$ is locally maximin,
it follows that:

$\sum_{u \in S^*} \f^{\vecp}(u) = \f(S^*)
\leq \rho_1 \cdot \g(S^*) = \rho_1 \cdot \sum_{u \in S^*} \g^{\vecq}(u)$,
where the equality follows because
$p_\sigma > 0$ or $q_\sigma > 0$ implies that $S^* \prec_\sigma V \setminus S^*$,
and the inequality follows because $\rho_1$ is the maximum density.

Moreover, by the definition of $S^*$,
for all $v \in V \setminus S^*$,
$\f^{\vecp}(v) < \frac{\f(S^*)}{\g(S^*)} \cdot \g^{\vecq}(v) \leq
\rho_1 \cdot \g^{\vecq}(v) $.

This immediately gives the following contradiction:

$$\rho_1 = \frac{\sum_{u \in S^*} \f^{\vecp}(u) + \sum_{v \in V \setminus S^*} \f^{\vecp}(v)}{\sum_{u \in S^*} \g^{\vecq}(u) + \sum_{v \in V \setminus S^*} \g^{\vecq}(v)}
< 
\frac{\rho_1 \cdot \sum_{u \in S^*} \g^{\vecq}(u) + \rho_1 \cdot\sum_{v \in V \setminus S^*} \g^{\vecq}(v)}{\sum_{u \in S^*} \g^{\vecq}(u) + \sum_{v \in V \setminus S^*} \g^{\vecq}(v)} = \rho_1.$$

\ignore{
Observe that for all $p_\sigma > 0$ or $q_\sigma > 0$,
we have $V^{(p,q)} \prec_\sigma V \setminus V^{(p,q)}$.
Since $\g^{\vecq}(u) = 0$ for all $u \notin V^{(p,q)}$,
it follows that $\g(V | V^{(p,q)}) = 0$.
Therefore, the objective function is $\Phi(\vecp, \vecq) = \g(V) \cdot \vartheta(\rho_1)$.

\hubert{Continue from here....}
}

\noindent \textbf{Induction Step} $k > 1$,
i.e., the density decomposition $V = \cup_{i=1}^k S_i$
has at least 2 parts. 
As before, suppose a locally maximin solution $(\vecp, \vecq) \in \Delta(\mcal{S}_V)^2$
induces
$S^{*}:=\argmax_{u\in V} \rho^{(\vecp, \vecq)}(u)$.
We first show that $S^* = S_1$.

Since $(\vecp, \vecq) $ is locally maximin,
$S^* \prec_{(\vecp, \vecq) } V \setminus S^*$.  Hence,
for all $u \in S^*$, $\rho^{(\vecp, \vecq) }(u) = \rho^* := \frac{\f(S^*)}{\g(S^*)}$.

It follows that for all $v \in V$,

$\rho^{(\vecp, \vecq) }(v) \leq \rho^* \leq \rho_1 := \frac{\f(S_1)}{\g(S_1)}$,
where the first inequality follows
from the definition of $S^*$ and the second inequality
follows because $S_1$ is the densest subset.

We will next transform $(\vecp, \vecq)$ into
another solution $(\widetilde{\vecp}, \widetilde{\vecq}) $ as follows.
Suppose $p_\sigma > 0$ or $q_\sigma > 0$ and it does not hold that $S_1 \prec_\sigma V \setminus S_1$.
Then, we perform ``stable sort'' to bring elements in $S_1$ to the front
(but not changing the relative order of elements within $S_1$ or $V \setminus S_1$)
to produce the resulting permutation $\sigma'$.
Then, any positive weight $p_\sigma > 0$ or $q_\sigma > 0$ will be transferred to 
$p_{\sigma'}$ or $q_{\sigma'}$, correspondingly.

Observe that for $u \in S_1$, $\f^{\sigma'}(u) \leq \f^\sigma(u)$
and $\g^{\sigma'}(u) \geq \g^{\sigma}(u)$.

Hence, after the transformation,
the new solution $(\widetilde{\vecp}, \widetilde{\vecq})$ satisfies that:

\begin{compactitem}

\item For all $u \in S_1$,

\begin{equation} \label{eq:u_S1}
\f^{\widetilde{\vecp}}(u) \leq \f^{\vecp}(u) \leq \rho^* \cdot \g^{\vecq}(u) \leq
\rho^* \cdot \g^{\widetilde{\vecq}}(u).
\end{equation}

\item Moreover, $S_1 \prec_{(\widetilde{\vecp}, \widetilde{\vecq})} V \setminus S_1$.
This implies that $\sum_{u \in S_1} \f^{\widetilde{\vecp}}(u) = \f(S_1)$
and $\sum_{u \in S_1} \g^{\widetilde{\vecq}}(u) = \g(S_1)$.

\end{compactitem}

Hence, summing (\ref{eq:u_S1}) over $u \in S_1$,
we have $\f(S_1) \leq \rho^* \cdot \g(S_1)$,
which implies that $\rho_1 = \frac{\f(S_1)}{\g(S_1)} \leq \rho^*$.
Therefore, $\rho_1 = \rho^*$, and equality holds in (\ref{eq:u_S1})
for all $u \in S_1$, which implies that $S_1 \subseteq S^*$.

However, since $S^*$ has the maximum density $\rho_1 = \rho^* = \frac{\f(S^*)}{\g(S^*)}$,
it must be the case that $S^* \subseteq S_1$, because $S_1$ is the maximal densest subset.
Hence,
we have for all $u \in S_1$,
$\rho^{(\vecp, \vecq)}(u) = \rho_1$.

We next consider a subinstance $(\widehat{V} := V \setminus S_1; \f(\cdot |S_1), \g(\cdot |S_1))$,
which is the next step in Definition~\ref{defn:density_decomp} of the density decomposition.
Moreover, the solution $(\vecp, \vecq)$ naturally induces $(\widehat{\vecp}, \widehat{\vecq}) \in \Delta(\mcal{S}_{\widehat{V}})^2$.
Specifically, $\widehat{p}_\tau$ is the aggregation of weights from all $p_{\sigma}$ such that 
$\tau \in \mcal{S}_{\widehat{V}}$ is a suffix of $\sigma \in \mcal{S}_{V}$;
$\widehat{q}_\tau$ is defined similarly.

Observe that for all $u \in \widehat{V}$,
$\rho^{(\vecp, \vecq)}(u) = \rho^{(\widehat{\vecp}, \widehat{\vecq})}(u)$;
moreover, $(\widehat{\vecp}, \widehat{\vecq})$ is also locally maximin.
Therefore, applying the induction hypothesis to the subinstance finishes the proof.

Recall that by definition $S^* = S_1$ are the coordinates
with the largest values in $\rho^* \in \R^V$.
Hence, the same induction proof shows that $\rho_1 > \rho_2 > \ldots > \rho_k$.

\noindent \textbf{Same Objective Value.}
Since $(\vecp, \vecq)$ is locally maximin,
it follows that with the density decomposition $V = \cup_{i=1}^k S_i$,
we have $S_1 \prec_{(\vecp, \vecq)} S_2  \prec_{(\vecp, \vecq)} \cdots \prec_{(\vecp, \vecq)} S_k$.

Therefore, we have $\sum_{u \in S_i} \g^{\vecq}(u) = \g(S_i | S_{<i})$.
This means that 

$\Phi(\vecp, \vecq) = \sum_{i=1}^k \g(S_i | S_{<i}) \cdot \vartheta(\frac{\f(S_i | S_{<i})}{\g(S_i | S_{<i})})$.

\end{proof}

\begin{corollary}[Optimal Solution Gives Density Decomposition]
\label{cor:local_maximin}
Given an instance $(V; \f, \g)$ with strictly monotone $\g$,
an optimal solution to the convex program in Definition~\ref{defn:general_convex} can be attained with the equality constraint $\vecp = \vecq$.
Moreover, for strictly convex~$\vartheta$, the induced density vector satisfies $\rho^{\vecp} = \rho^* \in \R^V$, i.e., it
coincides with the density vector from the density decomposition $V = \cup_{i=1}^k S_i$ in Definition~\ref{defn:density_decomp}.

On the other hand, as long as $\vartheta$ is (not necessarily strictly) convex,
any locally maximin solution is also an optimal solution.
\end{corollary}

\begin{proof}
We first consider a general convex $\vartheta$ that is not necessarily strictly convex.
By Lemmas~\ref{lemma:maximin_opt} and~\ref{lemma:maximin_density},
any optimal solution $(\vecp, \vecq)$ to the convex program
has the same objective value given in 
Lemma~\ref{lemma:maximin_opt}.

Next, consider the convex program with the extra equality constraint,
which also has a locally maximin optimal solution by Lemma~\ref{lemma:maximin_opt}, which has the same objective value by Lemma~\ref{lemma:maximin_density}.
Therefore, adding the extra equality constraint $\vecp = \vecq$ does not
increase the value of optimal objective value of the convex program.

When $\vartheta$ is strictly convex, Lemma~\ref{lemma:maximin_opt}
states that any optimal solution induces the same density vector
as the density decomposition.
\end{proof}

\subsection{An Example to Illustrate Strict Monotonicity Assumption}
\label{sec:strict_example}

We give an example to illustrate the importance of Assumption~\ref{assume:g}.
The ground set is $V = \{a, w, b\}$.
The functions $\f$ and $\g$ are represented in the following table.

For instance, given the column for the permutation $\sigma_0 = (a w b)$
and the row for $\g^\sigma$, the corresponding entry $102$ means that:
$\g(\{a\}) = 1$, $\g(\{a, w\}) = 1 + 0$ and $\g(\{a, w, b\}) = 1 + 0 + 2$.

\[
\begin{array}{c|cccccc}
\sigma & awb & abw & wab & wba & bwa & baw \\
\hline
\f^\sigma & 1 0 1 & 1 0 1 & 0 1 1 & 0 1 1 & 0 1 1 & 0 1 1 \\
\g^\sigma & 1 0 2 & 1 2 0 & 1 0 2 & 1 2 0 & 2 1 0 & 2 1 0 \\
\end{array}
\]

\noindent \textbf{Density Decomposition.}
According to Definition~\ref{defn:density_decomp},
we have $S_1 = \{a, w\}$ and $S_2 = \{b\}$.
From the definition, the density of the element $w$ should be $1$. 

\noindent \textbf{Optimal Solution for Convex Program.}
The only optimal solution to the convex program in Definition~\ref{defn:general_convex}
is induced by $\sigma_0 = (awb)$, which is represented by $(x, y) = (\f^{\sigma_0}, \g^{\sigma_0}) \in \mcal{B}^{\geq}_{\f}
\times \mcal{B}^{\leq}_{\g}$ as follows:

\[
\begin{array}{c|ccc}
V & a & w & b  \\
\hline
x & 1 & 0 & 1  \\
y & 1 & 0 & 2 \\
\end{array}
\]

Note that in this case, the induced density of the element~$w$ is undefined
in $(\f^{\sigma_0}, \g^{\sigma_0})$.

\section{Application to Dual-Modular Combinatorial Contracts}
\label{sec:contracts}

As mentioned in the introduction,
given a dual-modular instance $(V; \f, \g)$,
and a contract parameter $\alpha \in [0,1]$,
the agent would like
to make a best response
$\max \{\alpha \cdot \f(S) - \g(S): S \subseteq V \}$.
We write $\gamma = \frac{1}{\alpha}$, and
will consider a wider range $\gamma \geq 0$.
We analyze the best response using two approaches:
(i) a strong duality relationship
with a convex program in Definition~\ref{defn:general_convex},
(ii) direct properties of the density decomposition
in Definition~\ref{defn:density_decomp}.

\noindent \textbf{Hockey-Stick Divergence.}
Recall that setting $\vartheta_\gamma(t) := \max\{t - \gamma, 0\}$
in Definition~\ref{defn:general_convex},
we denote $\mathsf{HS}_\gamma(x \| y) := \D_{\vartheta_\gamma}(x \| y)$.

\noindent \textbf{Density Decomposition.}
Recall that density decomposition in Definition~\ref{defn:density_decomp}
that gives the partition $V = \cup_{i=1}^k S_i$,
with strictly decreasing densities $\rho_i := \rho(S_i | S_{<i})$.

We restate Lemma~\ref{lemma:main_dual_contracts}, and give its proof.

\begin{lemma}[Strong Duality Between Best Response and Hockey-Stick]
\label{lemma:dual_contracts}
For $\gamma \geq 0$, any $S \subseteq V$ and allocation~$(x, y) \in \Bfg$, we have:
$\f(S) - \gamma \cdot \g(S) \leq \mathsf{HS}_\gamma(x \| y)$.

If $i$ is the largest index such that $\rho_i \geq \gamma$,
then equality is attained by the subset $S_{\leq i}$ and any locally maximin~$(x, y)$.
\end{lemma}

\begin{proof}
By Fact~\ref{fact:local_represent},
we consider some locally maximin solution $(\vecp, \vecq) \in \Delta(\mathcal{S}_V)^2$ and fix some $S \subseteq V$.
For each permutation $\sigma \in \mathcal{S}_V$,
consider the transformed permutation $\overline{\sigma}$,
where all elements in $S$ are brought to the beginning,
but the relative order of elements within $S$ or $V \setminus S$ remains
unchanged.
Note that because of the dual-modularity of the functions $(\f, \g)$,
we have:

 For all $u \in S$, 
$\f^\sigma(u) \geq \f^{\overline{\sigma}}(u)$
and
$\g^\sigma(u) \leq \g^{\overline{\sigma}}(u)$.

Moreover, $\sum_{u \in S} \f^{\overline{\sigma}}(u) = \f(S)$
and $\sum_{u \in S} \g^{\overline{\sigma}}(u) = \g(S)$.

When we apply this transformation analogously to $(\vecp, \vecq)$,
in each permutation distribution, we transfer all the weight
from permutation $\sigma$ to $\overline{\sigma}$ to produce
the transformed pair $(\overline{\vecp}, \overline{\vecq})$.

Similarly, we have for all $u \in S$,
$\f^\vecp(u) \geq \f^{\overline{\vecp}}(u)$
and
$\g^\vecq(u) \leq \g^{\overline{\vecq}}(u)$; moreover, we have:

$\sum_{u \in S} \f^{\overline{\vecp}}(u) = \f(S)$
and $\sum_{u \in S} \g^{\overline{\vecq}}(u) = \g(S)$.

Therefore, we have

\begin{align*}
\mathsf{HS}_\gamma(\f^\vecp \| \g^\vecq) 
& = \sum_{u \in V} \g^\vecq(u) \cdot \max\{\frac{\f^\vecp(u)}{\g^\vecq(u)} - \gamma, 0\} \geq \sum_{u \in S} (\f^\vecp(u) - \gamma \cdot \g^\vecq(u)) \\
& \geq \sum_{u \in S} (\f^{\overline\vecp}(u) - \gamma \cdot \g^{\overline\vecq}(u))
= \f(S) - \gamma \cdot \g(S),
\end{align*}

as required.

For equality, suppose $i$ is the largest index
in the density decomposition such that $\rho_i \geq \gamma$.
Note that Lemma~\ref{lemma:maximin_density} states that
any locally maximin solution has the objective value:

$\sum_{j=1}^k \g(S_j|S_{< j}) \cdot \max\{\rho_j - \gamma , 0\}
= \sum_{j=1}^i \g(S_j|S_{< j}) \cdot (\frac{\f(S_j|S_{< j})}{\g(S_j|S_{< j})} - \gamma) = \f(S_{\leq i}) - \gamma \cdot \g(S_{\leq i})$,

as required.
\end{proof}

\noindent \textbf{Smallest Collection of Critical Values.}
As mentioned in Theorem~\ref{th:main_contracts}, the next lemma shows that the densities in the decomposition are the critical values. In this second approach,
we will use the properties of the density decomposition directly.

\begin{lemma}[Unique Best Response]
Suppose $\g$ is strictly monotone, and
 $\rho_i > \gamma > \rho_{i+1}$.
Then, the unique maximizer
$\max \{\f(S) - \gamma \cdot \g(S): S \subseteq V \}$
is $S_{\leq i}$.
\end{lemma}

\begin{proof}
Suppose $S \subseteq V$ does not equal $S_{\leq i}$.
Then, we show that the objective $\f(S) - \gamma \cdot \g(S)$
can be strictly increased.

\noindent \textbf{Case 1:}
$S$ does not include $S_{\leq i}$.  Suppose $j \leq i$
is the smallest index such that $S$ does not contain $S_j$.  Denote 
$R := S_j \setminus S \neq \emptyset$.

From the construction of the density decomposition,
we have: 

$\f(S \cap S_j|S_{<j}) \leq \rho_j \cdot \g(S \cap S_j|S_{<j})$
and $\f(S_j|S_{<j}) = \rho_j \cdot \g(S_j|S_{<j})$, which implies that

$\f(R|S_{<j}) \geq \rho_j \cdot \g(R|S_{<j}) > \gamma \cdot \g(R|S_{<j})$,
where the last strict inequality follows because $R$ is non-empty
and $\g$ is strictly monotone.

Because of the dual-modularity of $(\f, \g)$,
we have $\f(R|S) \geq \f(R|S_{<j})$ and $\g(R|S) \leq \g(R|S_{<j})$.

Therefore, by adding $R$ to $S$,
the objective can increase by $\f(R|S) - \gamma \cdot \g(R|S) \geq 
\f(R|S_{<j}) - \gamma \cdot \g(R|_{<j}) > 0$.

\noindent \textbf{Case 2:} $S$ includes some element in $S_{\geq i+1}$.
Suppose $j \geq i+1$ is the largest index such that $R := S \cap S_j \neq \emptyset$.
We show that the marginal contribution of $R$ to the objective function given $S \setminus R$ is negative.

The property of the density decomposition gives
that $\f(R|S_{<j}) \leq \rho_j \cdot \g(R|S_{<j}) < \gamma \cdot \g(R|S_{<j})$.

Again, by dual-modularity,
$\f(R|S \setminus R) \leq \f(R|S_{<j})$
and $\g(R|S \setminus R) \geq \g(R|S_{<j})$.

Therefore, we have
$\f(R|S \setminus R) - \gamma \cdot \g(R|S \setminus R)
\leq \f(R|S_{<j}) - \gamma \cdot \g(R|S_{<j}) < 0$.

Hence, removing $R$ from $S$ will strictly increase the objective function.

\end{proof}

\noindent \textbf{Approximation for Contracts Problem.}
Unfortunately, it seems that there is no meaningful approximation guarantee for the contracts problem.
Consider some instance $(V; \f, \g)$
with the decomposition $V = S_1 \cup S_2$, where $n_1 = |S_1|$ and $n_2 = |S_2|$.
The function $\g$ is linear and $\g(u) = 1$ for $u \in S_1$
and $\g(v) = 10 n_1$ for $v \in S_2$.
Moreover, $\f$ is constructed such that 
$\f(S_1) = 2 n_1$ and $\f(S) = 0$ for all $S \subsetneq S_1$;
$\f(S_2|S_1) = 10 n_1 n_2$, and
$\f(A|S_1) = \f(S_2|S) = 0$, for all $A \subsetneq S_2$ and $S \subsetneq S_1$.
Hence, in the density vector $\rho^*$, the elements in $S_1$ have
densities 2 and the elements in $S_2$ have densities 1.

Given $1 < \gamma < 2$,
the best response for the agent is $S_1$
and $\f(S_1) - \gamma \cdot \g(S_1) = (2 - \gamma) \cdot n_1 > 0$.
However, any other non-empty $S \neq S_1$ would lead to
$\f(S) - \gamma \cdot \g(S) < 0$.

\section{Iterative Algorithms via Frank-Wolfe}
\label{sec:FW}

As mentioned in Section~\ref{sec:prelim},
Frank-Wolfe is an iterative procedure that can be applied to
a convex program to achieve approximation guarantees.
Specifically,
we will apply it to the convex program
in Definition~\ref{defn:general_convex}
with the equality constraint $\vecp = \vecq$
and the objective function
$\Phi(\vecp) = \D_\vartheta(\f^\vecp \| \g^\vecp)$.

\noindent \textbf{Compact Solution Representation.}
Even though the convex program $\Phi(\vecp)$ in Definition~\ref{defn:general_convex}
has an exponential number of variables in $\vecp$,
it suffices to maintain the corresponding $\f^{\vecp}$ and 
$\g^{\vecp} \in \R^V$ in an iterative algorithm.
We give the generalization of Algorithm~\ref{alg:abstract} as follows.

\begin{algorithm}[H]
\caption{General Abstract Iterative Procedure}\label{alg:general_abstract}
\begin{algorithmic}[1]
\State \textbf{Input:} Supermodular $\f : 2^V \to \R$; submodular $\g: 2^V \to \R$;
step size $\gamma: \Z_{\geq 0} \to [0,1]$;
number $T$ of iterations.

\State    \text{\textbf{Initialize}:} Pick an arbitrary permutation $\sigma_0 \in \mcal{S}_V$;
set $(x^{(0)}, y^{(0)}) \gets (\f^{\sigma_0}, \g^{\sigma_0}) \in \R^V \times \R^V$.

\For{ $k \gets 0 \text{ to } T-1$}

		
		\State $\sigma_{k+1} \gets \textsc{Abstract-Perm}(\f, \g, \gamma_k, x^{(k)}, y^{(k)})$
		\label{ln:abstract_perm}
		
		\State $(c^{(k+1)}, d^{(k+1)}) \gets (\f^{\sigma_{k+1}}, \g^{\sigma_{k+1}})$

     \State $(x^{(k+1)}, y^{(k+1)}) \gets (1 - \gamma_k) \cdot (x^{(k)}, y^{(k)})+\gamma_k \cdot (c^{(k+1)}, d^{(k+1)})$ \label{ln:update2}
 
\EndFor
\State \text{\textbf{return: $(x^{(T)}, y^{(T)}) \in \R^V \times \R^V$} }  
\end{algorithmic}
\end{algorithm}

\noindent \textbf{Intuition.}
To instantiate the general abstract procedure in Algorithm~\ref{alg:general_abstract}
as Frank-Wolfe, in addition to choosing the rate $\gamma_k = \frac{2}{k+2}$,
the procedure \textsc{Abstract-Perm} in line~\ref{ln:abstract_perm}
is supposed to take an (implicit) distribution $\vecp^{(k)} \in \Delta(\mcal{S}_V)$
and returns some permutation distribution in:

$$\min_{\vecp' \in \Delta(\mcal{S}_V)} \langle \vecp', \nabla \Phi(\vecp^{(k)}) \rangle.$$

However, because the domain $\Delta(\mcal{S}_V)$ is a probability distribution,
the minimum can be achieved by a distribution whose probability mass is concentrated
on exactly one permutation.  In other words, for any $\vecp \in \Delta(\mcal{S}_V)$,  we have:

\begin{equation} \label{eq:min_perm}
\min_{\vecp' \in \Delta(\mcal{S}_V)} \langle \vecp', \nabla \Phi(\vecp) \rangle
= \min_{\sigma \in \mcal{S}_V} \frac{\partial{\Phi(\vecp)}}{\partial{p_\sigma}}.
\end{equation}

Hence, the procedure \textsc{Abstract-Perm} will
return a permutation in $\mcal{S}_V$ that achieves the minimum.
Lemma~\ref{lemma:greedy_perm} shows that such a permutation can be 
achieved by sorting the elements $v$ in $V$
according to their induced densities $\rho^{\vecp}(v)$
in non-increasing order. We will use the following fact in the proof.

\begin{fact}[Dot-Product Optimization]
\label{fact:dotproduct}
Suppose $h \in \R^V$. Then, the following
dot-product optimization problems can be achieved by sorting the elements
$v \in V$ according to the values $h(v)$ as follows:

\begin{compactitem}
\item Suppose $\f$ is supermodular. Then,
the optimal value $\min_{\sigma \in \mcal{S}_V} \langle \f^\sigma, h \rangle$
can be attained by a permutation from sorting $v \in V$ in non-increasing order of $h(v)$.

\item Suppose $\g$ is submodular. Then,
the optimal  value $\min_{\sigma \in \mcal{S}_V} \langle \g^\sigma, h \rangle$
can be attained by a permutation from
sorting $v \in V$ in non-decreasing order of $h(v)$.
\end{compactitem}
\end{fact}

\begin{lemma}[Greedy Permutation Achieves Minimum Dot-Product with Gradient]
\label{lemma:greedy_perm}
Suppose $\vartheta$ is a differentiable convex function
in the objective function $\Phi(\vecp):= \D_\vartheta(\f^\vecp \| \g^\vecp)$.
Then, for any $\vecp \in \Delta(\mcal{S}_V)$,
a minimum permutation~$\sigma$ in equation~(\ref{eq:min_perm})
can be achieved by 
sorting the elements $v$ in $V$
according to their induced densities $\rho^{\vecp}_v= \frac{\f^\vecp_v}{\g^\vecp_v}$
in non-increasing order.
\end{lemma}

\begin{proof}
For $\sigma \in \mcal{S}_V$,
we derive:

\begin{align} \label{eq:grad}
\frac{\partial{\Phi}(\vecp)}{\partial{p_\sigma}}
=\sum_{v\in V}\f^{\sigma}_v\cdot\vartheta'(\rho_{v}^{\vecp}) +\sum_{v \in V}\g_{v}^{\sigma}\cdot(\vartheta(\rho_{v}^{\vecp})- \rho_{v}^{\vecp} \cdot \vartheta'(\rho_{v}^{\vecp})).
\end{align}

We next apply Fact~\ref{fact:dotproduct} as follows.

\begin{compactitem}

\item In the first term of~(\ref{eq:grad}),
observe that the function $t \mapsto \vartheta'(t)$ is non-decreasing,
because $\vartheta$ is convex.

\item In the second term of~(\ref{eq:grad}),
the function $\zeta(t) := \vartheta(t) - t \cdot \vartheta'(t)$ is non-increasing.
This is because, for $0 \leq x < y$,
$\vartheta(y) - \vartheta(x) \leq  (y - x) \vartheta'(y)  \leq
y \vartheta'(y) - x \vartheta'(x)$.
Note that the first inequality follows from the mean value theorem and the convexity
of $\vartheta$.

\end{compactitem}

Therefore, by Fact~\ref{fact:dotproduct},
a permutation~$\sigma$ that minimizes~(\ref{eq:grad})
can be achieved by sorting the elements~$v \in V$
in non-increasing order according to the values $\rho^{\vecp}_v$.
\end{proof}

\subsection{Convergence Analysis of Frank-Wolfe}

We summarize the convergence analysis of Frank-Wolfe.

\begin{definition}[Curvature Constant~\cite{DBLP:conf/icml/Jaggi13,DBLP:conf/esa/HarbQC23}]
Let $\mathcal{F}: \mcal{X} \rightarrow \mathbb{R}$ be a convex and (twicely) differentiable function, where $\mcal{X}$ is a compact convex set
in some Euclidean space.

The curvature constant $C_{\mathcal{F}}$
has the following upper bound:

\begin{align*}
C_{\mathcal{F}} & :=\sup\limits_{s',s''\in \mcal{X},\lambda\in[0,1],s=(1-\lambda)s'+\lambda s''}\frac{2}{\lambda^2}(\mathcal{F}(s)-\mathcal{F}(s')-\langle s-s',\nabla \mathcal{F}(s')\rangle) \\
& \leq \mathsf{Diam}(\mcal{X})^2 \cdot \|\nabla^2 \mcal{F}\|,
\end{align*}

where $\mathsf{Diam}(\mcal{X}) = \max_{s, s' \in \mcal{X}} \|s - s'\|_2$
and $\|\nabla^2 \mcal{F}\|$ is the spectral norm
of the Hessian of $\mcal{F}$.

%
%
\end{definition}

\begin{fact}[Convergence Rates of Frank-Wolfe~\cite{DBLP:conf/icml/Jaggi13,DBLP:conf/esa/HarbQC23}]
Initialize $x^{(0)} \in \mcal{X}$ arbitrarily
with the update $x^{(k+1)} \gets (1 - \gamma_k) \cdot x^{(k)}
+ \gamma_k \cdot s^{(k)}$
and $s^{(k)} = \argmin_{s \in \mcal{X}} \langle s, \nabla \mcal{F}(x^{(k)})\rangle$. 

For $T \geq 1$, we have the following upper bounds
$\mcal{F}(x^{(T)}) - \mcal{F}(x^*)$ on the difference
from the optimal objective value:

\begin{compactitem}
\item For $\gamma_k := \frac{2}{k+2}$, the upper bound is $C_\mcal{F} \cdot \frac{2}{T+2}$.

\item For $\gamma_k := \frac{1}{k+1}$, the upper bound
is $C_\mcal{F} \cdot O(\frac{\log T}{T})$.

\end{compactitem}

\end{fact}

\noindent \textbf{Applying Frank-Wolfe to the General Convex Program
in Definition~\ref{defn:general_convex}.}
Recall that there are two equivalent perspectives.
We can choose the domain as either (i) vectors in polytopes~$\mcal{B}^{\geq}_\f \times \mcal{B}^{\leq}_\g$,
or (ii)~permutation distribution~$\Delta(\mcal{S}_V)$.  In Section~\ref{sec:technical},
we see that the permutation distribution is convenient.  However, we shall see that
the polytope perspective is going to give a better bound on $C_\Phi$.

We shall consider several choices of the convex function~$\vartheta$:

\begin{enumerate}[(i)]

\item $\vartheta(t) := t^2$;

\item $\vartheta(t) := t \log t$;

\item $\vartheta(t) := - \log t$.

\end{enumerate}

We first consider case~(i), where the 
objective function
$\Phi: \R^V \times \R^V \rightarrow \R$ becomes:

$$\Phi(x ,y) = \sum\limits_{u\in V} y_u \cdot (\frac{x_u}{y_u})^2 = \sum_{u \in V} \frac{x_u^2}{y_u}.$$

\noindent \textbf{Parameters.}
We will express the bounds in terms of the parameters:
$\f_V := \f(V)$, $\g_V := \g(V)$, $\f_{\max} := \max_{\sigma, u} \f^{\sigma}_u$
and $\f_{\min} := \min_{\sigma, u} \f^\sigma_u$;
define $\g_{\max}$ and $\g_{\min}$ similarly.
For notational convenience, we consider the normalization $\f_V = \g_V = 1$,
and use the upper bounds $\f_{\max}, \g_{\max} \leq 1$.

\begin{lemma}[Diameter Bound]
$\mathsf{Diam}(\mcal{B}^{\geq}_\f \times \mcal{B}^{\leq}_\g)^2 \leq 2(\f_V^2 + \g_V^2) = 4$.
\end{lemma}

\begin{proof}
Observe that 
$\mathsf{Diam}(\mcal{B}^{\geq}_\f \times \mcal{B}^{\leq}_\g)^2
= \mathsf{Diam}(\mcal{B}^{\geq}_\f)^2 + \mathsf{Diam}(\mcal{B}^{\leq}_\g)^2$.

For $x, x' \in \mcal{B}^{\geq}_\f$,
observe that $\sum_{u \in V} x_u = \sum_{u \in V} x'_u = \f_V$.
Hence, $\|x - x'\|_2^2 \leq 2 \f_V^2$.

Similarly, for $y, y' \in \mcal{B}^{\leq}_\g$,
$\|y - y'\|_2^2 \leq 2 \g_V^2$.
\end{proof}

\begin{lemma}[Spectral Norm of Hessian $\vartheta(t)=t^2$]
For $(x,y) \in \mcal{B}^{\geq}_\f \times \mcal{B}^{\leq}_\g$,

$\| \nabla^2 \Phi(x,y) \| = \max_{u \in V} \frac{2(x_u^2 + y_u^2)}{y_u^3}
\leq \frac{4}{\g_{\min}^3}.$

\end{lemma}

\begin{proof}
We consider the non-zero partial derivatives:

$$\frac{\partial{\Phi}(x, y)}{\partial{x_u}}=\frac{2x_u}{y_u}, \frac{\partial{\Phi}}{\partial{y_u}}(x,y)=-\frac{x^2_u}{y^2_u}$$

$$
\frac{\partial^2{\Phi}(x,y)}{\partial{x^2_u}} = \frac{2}{y_u},
\frac{\partial^2{\Phi}(x, y)}{\partial{y^2_u}}=\frac{2x^2_u}{y^3_u},
\frac{\partial^2{\Phi}(x,y)}{\partial{x_u}\partial{y_u}}=-\frac{2x_u}{y^2_u}.$$

\ignore{
\begin{align*}
\frac{\partial{\Phi}}{\partial{x_u}}(x,y)&=\frac{2x_u}{y_u}\\
\frac{\partial{\Phi}}{\partial{y_u}}(x,y)&=-\frac{x^2_u}{y^2_u}\\
\frac{\partial^2{\Phi}}{\partial{x_u}\partial{y_u}}(x,y)&=-\frac{2x_u}{y^2_u}\\
\frac{\partial^2{\Phi}}{\partial{x^2_u}}(x,y)&=\frac{2}{y_u}\\
\frac{\partial^2{\Phi}}{\partial{y^2_u}}(x,y)&=\frac{2x^2_u}{y^3_u}
\end{align*}
}

Note that the Hessian matrix has the form 
$\nabla \Phi(x,y) = \begin{pmatrix}
{D}_1&{D}_2\\
D_2&D_3
\end{pmatrix}$, where $D_i$ are $|V|\times |V|$ diagonal matrices for $i=1,2,3$. Then, the eigenvalues of the Hessian matrix are the roots of the polynomial in $\lambda$: 

$$\lambda^{|V|}\cdot\Pi_{u\in V}(\lambda-(\frac{2}{y_u}+\frac{2x_u^2}{y^3_u}))=0.$$ 

Thus, we have $$\|\nabla^{2}\Phi(x,y)\|_2=\max\limits_{u\in V}\frac{2(x_u^2+y_u^2)}{y_u^3},$$

as required.

\ignore{
Assume $F_{\max}=\max\limits_{u,\sigma}\f^{\sigma}_u$, $F_{\min}=\min\limits_{u,\sigma}\f^{\sigma}_u$, $G_{\max}=\max\limits_{u,\sigma}\g^{\sigma}_u$, $G_{\min}=\min\limits_{u,\sigma}\g^{\sigma}_u$.
\begin{align*}
C_{\Phi}&\leq (diam(\mcal{B}^{\geq}_\f\times\mcal{B}^{\leq}_\g))^2\cdot \sup\limits_{(x,y)\in \mcal{B}^{\geq}_\f\times\mcal{B}^{\leq}_\g}\|\nabla^2\Phi(x,y)\|_2\\
&\leq (diam(\mcal{B}^{\geq}_\f)^2+diam(\mcal{B}^{\leq}_\g)^2)\cdot \sup\limits_{(x,y)\in \mcal{B}^{\geq}_\f\times\mcal{B}^{\leq}_\g}\|\nabla^2\Phi(x,y)\|_2\\
&\leq 2\cdot|V|\cdot((F_{\max}-F_{\min})^2+(G_{\max}-G_{\min})^2)\cdot (F_{\max}^2+G_{\max}^2)\cdot G_{\min}^{-3}\\
&\leq 2\cdot|V|\cdot(F^2_{\max}+F^2_{\min}+G^2_{\max}+G^2_{\min})\cdot (F_{\max}^2+G_{\max}^2)\cdot G_{\min}^{-3}\\
&\leq 4\cdot|V|\cdot G_{\min}^{-3}\cdot (F_{\max}^2+G_{\max}^2)^2
\end{align*}
}
\end{proof}

\subsection{Approximation Analysis for Density Vector}

\noindent \textbf{Deriving Approximation Guarantee for Density Vector from
Guarantee for Objective Function.}
Recall that we consider
a feasible solution $z = (x, y) \in \mcal{B}^{\geq}_\f\times\mcal{B}^{\leq}_\g$
with the objective function
$\Phi(z) := \sum_{u \in V} \frac{x_u^2}{y_u}$.
Let $z^* = (x^*, y^*)$ be some optimal solution;
observe that $z^*$ is not unique, but any optimal solution
will induce a unique density vector $\rho^*$.

Given an upper bound on $\mcal{E}(z) := \Phi(z) - \Phi(z^*)$,
we would like to derive some approximation guarantee
for the induced density vector $\rho \in \R^V$ given by $\rho_u := \frac{x_u}{y_u}$ for $u \in V$.

\noindent \textbf{Overall Strategy.}
Consider $\theta \in [0,1]$
and let $z(\theta) := z^* + \theta \cdot (z - z^*)$.
Define the function $\Gamma: [0,1] \rightarrow \R$
by $\Gamma(\theta) := \Phi(z(\theta))$.
We shall derive an appropriate uniform lower bound $\Gamma''(\theta) \geq L$ for all $\theta \in [0,1]$,
from which we can conclude

$$\mcal{E}(z) = \Phi(z) - \Phi(z^*) = \Gamma(1) - \Gamma(0) \geq \frac{L}{2}.$$

\begin{lemma}[$\vartheta(t)=t^2$]
For $\theta \in [0,1]$ and $u \in V$,
$\Gamma''(\theta) \geq 
2 \g_{\min}^2 \cdot \|\rho - \rho^*\|_2^2
\geq 2 (\f_{\min} \cdot \g_{\min})^2 \cdot \epsilon_u^2$,
where $\epsilon_u := \frac{\rho_u - \rho^*_u}{\rho^*_u}$.
\end{lemma}

\begin{proof}
Denote $x(\theta) := x^* + \theta \cdot (x - x^*) \in \R^V$
and $y(\theta) := y^* + \theta \cdot (y - y^*) \in \R^V$.

Moreover, denote
$\rho(\theta) \in \R^V$ by
$\rho_u(\theta) := \frac{x_u(\theta)}{y_u(\theta)}$ for $u \in V$.
We have:

$\rho_u'(\theta) = y_u(\theta)^{-1} \cdot \{(x_u - x^*_u) - \rho_u(\theta) \cdot (y_u - y^*_u)\}.$

Then, we have:

$\Gamma'(\theta) = \sum_{u \in V} \{2 \rho_u(\theta) \cdot (x_u - x^*_u) -
\rho_u(\theta)^2 \cdot (y_u - y^*_u)\}$.

Therefore, we have:

\begin{align*}
\Gamma''(\theta) & = \sum_{u \in V} \frac{2}{y_u(\theta)} \cdot
\{(x_u - x^*_u) - \rho_u(\theta) \cdot (y_u - y^*_u)\}^2 \\
& \geq \frac{2}{\g_{\max}} \sum_{u \in V}  \{(x_u - x^*_u) - \rho_u(\theta) \cdot (y_u - y^*_u)\}^2.
\end{align*}

Denote $\epsilon_u(\theta) := \frac{\rho_u(\theta) - \rho^*_u}{\rho^*_u}$,
where $\epsilon_u = \epsilon_u(1)$.
Observe that $\epsilon_u \cdot \epsilon_u(\theta) \geq 0$ and
and $|\epsilon_u(\theta)| \leq |\epsilon_u|$.

Note that each term in the above sum contains the following expression:

$|(x_u - x^*_u) - \rho_u(\theta) \cdot (y_u - y^*_u)|
= |(\epsilon_u - \epsilon_u(\theta)) \cdot \rho^*_u y_u + \epsilon_u(\theta) \cdot \rho^*_u y^*_u|
\geq |\epsilon_u| \cdot \rho^*_u \cdot \g_{\min} \geq |\epsilon_u| \cdot \frac{\f_{\min} \cdot \g_{\min}}{\g_{\max}}.$

Therefore, we have:

$$\Gamma''(\theta) \geq 
\frac{2 \g_{\min}^2}{\g_{\max}} \sum_{u \in V}  \epsilon_u^2 \cdot (\rho^*_u)^2
= \frac{2 \g_{\min}^2}{\g_{\max}} \cdot \|\rho - \rho^*\|_2^2
\geq \frac{2 (\f_{\min}\cdot\g_{\min})^2}{\g_{\max}^3} \cdot \epsilon_u^2,$$

for all $u \in V$, as required.

\end{proof}

\begin{corollary}[Approximation Guarantees for Density Vector $\vartheta(t)=t^2$]
After $T$ steps of Frank-Wolfe,
the solution $z^{(T)} \in \mcal{B}^{\geq}_\f\times\mcal{B}^{\leq}_\g$
induces a density vector $\rho^{(T)} \in \R^V$ with the following approximation 
guarantees:

\begin{compactitem}

\item \emph{Absolute error}.

$$\| \rho^{(T)} - \rho^* \|_2 \leq \frac{O(1)}{\g_{\min}^{2.5} \cdot \sqrt{T+2}}.$$

\item \emph{Coordinate-wise Multiplicative error.}
For all $u \in V$,

$$\left|\frac{\rho^{(T)}_u - \rho^*_u}{\rho^*_u} \right| \leq \frac{O(1)}{\f_{\min} \cdot \g_{\min}^{2.5} \cdot \sqrt{T+2}}.$$

\end{compactitem}

\end{corollary}

\subsection{Repeating the Analysis for Choosing $\vartheta(t) := t \log t$}

\begin{lemma}[Spectral Norm of Hessian $\vartheta(t)=t\log{t}$]
For $(x,y) \in \mcal{B}^{\geq}_\f \times \mcal{B}^{\leq}_\g$,

$\| \nabla^2 \Phi(x,y) \| = \max_{u \in V} (\frac{1}{x_u}+\frac{1}{y_u^2})
\leq \frac{1}{\g^2_{\min}}+\frac{1}{\f_{\min}}.$

\end{lemma}

\begin{proof}

$\Phi(x,y)=\sum\limits_{u\in V}x_u\cdot \log{\frac{x_u}{y_u}}$

We consider the non-zero partial derivatives:

$$\frac{\partial{\Phi}(x, y)}{\partial{x_u}}=\log{\frac{x_u}{y_u}}+1, \frac{\partial{\Phi}}{\partial{y_u}}(x,y)=-\frac{x_u}{y_u}$$

$$
\frac{\partial^2{\Phi}(x,y)}{\partial{x^2_u}} = \frac{1}{x_u},
\frac{\partial^2{\Phi}(x, y)}{\partial{y^2_u}}=\frac{x_u}{y^2_u},
\frac{\partial^2{\Phi}(x,y)}{\partial{x_u}\partial{y_u}}=-\frac{1}{y_u}.$$

Note that the Hessian matrix has the form 
$\nabla \Phi(x,y) = \begin{pmatrix}
{D}_1&{D}_2\\
D_2&D_3
\end{pmatrix}$, where $D_i$ are $|V|\times |V|$ diagonal matrices for $i=1,2,3$. Then, the eigenvalues of the Hessian matrix are the roots of the polynomial in $\lambda$: 

$$\lambda^{|V|}\cdot\Pi_{u\in V}(\lambda-(\frac{1}{x_u}+\frac{x_u}{y^2_u}))=0.$$ 

Thus, we have $$\|\nabla^{2}\Phi(x,y)\|_2=\max_{u \in V} (\frac{1}{x_u}+\frac{1}{y_u^2}),$$

as required.

\end{proof}

\begin{lemma}[$\vartheta(t)=t\log{t}$]
For $\theta \in [0,1]$ and $u \in V$,
$\Gamma''(\theta) \geq 
\f_{\min}\cdot\g_{\min}^2 \cdot \|\rho - \rho^*\|_2^2
\geq { \f^3_{\min}\cdot\g^2_{\min}} \cdot \epsilon_u^2$,
where $\epsilon_u := \frac{\rho_u - \rho^*_u}{\rho^*_u}$.
\end{lemma}

\begin{proof}
Denote $x(\theta) := x^* + \theta \cdot (x - x^*) \in \R^V$
and $y(\theta) := y^* + \theta \cdot (y - y^*) \in \R^V$.

Moreover, denote
$\rho(\theta) \in \R^V$ by
$\rho_u(\theta) := \frac{x_u(\theta)}{y_u(\theta)}$ for $u \in V$.
We have:

$\rho_u'(\theta) = y_u(\theta)^{-1} \cdot \{(x_u - x^*_u) - \rho_u(\theta) \cdot (y_u - y^*_u)\}.$

Then, we have:

$\Gamma(\theta) = \sum_{u \in V} (x^{*}_u+\theta(x_u - x^*_u))\cdot \log{\rho_u(\theta)}$.

$\Gamma'(\theta) = \sum_{u \in V} (x_u - x^*_u)\cdot \log{\rho_u(\theta)}+(x_u - x^*_u)-(y_u - y^*_u)\cdot \rho_u(\theta)$.

Therefore, we have:

\begin{align*}
\Gamma''(\theta) & = \sum_{u \in V} x_{u}(\theta) \cdot
\{(x_u - x^*_u) - \rho_u(\theta) \cdot (y_u - y^*_u)\}^2 \\
& \geq \f_{\min} \sum_{u \in V}  \{(x_u - x^*_u) - \rho_u(\theta) \cdot (y_u - y^*_u)\}^2.
\end{align*}

Denote $\epsilon_u(\theta) := \frac{\rho_u(\theta) - \rho^*_u}{\rho^*_u}$,
where $\epsilon_u = \epsilon_u(1)$.
Observe that $\epsilon_u \cdot \epsilon_u(\theta) \geq 0$ and
and $|\epsilon_u(\theta)| \leq |\epsilon_u|$.

Note that each term in the above sum contains the following expression:

$|(x_u - x^*_u) - \rho_u(\theta) \cdot (y_u - y^*_u)|
= |(\epsilon_u - \epsilon_u(\theta)) \cdot \rho^*_u y_u + \epsilon_u(\theta) \cdot \rho^*_u y^*_u|
\geq |\epsilon_u| \cdot \rho^*_u \cdot \g_{\min} \geq |\epsilon_u| \cdot \frac{\f_{\min} \cdot \g_{\min}}{\g_{\max}}.$

Therefore, we have:

$$\Gamma''(\theta) \geq 
 \f_{\min}\cdot\g_{\min}^2\cdot\sum_{u \in V}  \epsilon_u^2 \cdot (\rho^*_u)^2
= \f_{\min}\cdot\g_{\min}^2\cdot \|\rho - \rho^*\|_2^2
\geq \frac{ \f^3_{\min}\cdot\g^2_{\min}}{\g_{\max}^2} \cdot \epsilon_u^2,$$

for all $u \in V$, as required.

\end{proof}

\begin{corollary}[Approximation Guarantees for Density Vector $\vartheta(t)=t\log{t}$]
After $T$ steps of Frank-Wolfe,
the solution $z^{(T)} \in \mcal{B}^{\geq}_\f\times\mcal{B}^{\leq}_\g$
induces a density vector $\rho^{(T)} \in \R^V$ with the following approximation 
guarantees:

\begin{compactitem}

\item \emph{Absolute error}.

$$\| \rho^{(T)} - \rho^* \|_2 \leq \frac{1}{\f_{\min}^{0.5}\cdot \g_{\min}}\cdot(\frac{1}{\g_{\min}^2}+\frac{1}{\f_{\min}})^{0.5}\cdot\frac{O(1)}{ \sqrt{T+2}}.$$

\item \emph{Coordinate-wise Multiplicative error.}
For all $u \in V$,

$$\left|\frac{\rho^{(T)}_u - \rho^*_u}{\rho^*_u} \right| \leq \frac{1}{\f_{\min}^{1.5}\cdot \g_{\min}}\cdot(\frac{1}{\g_{\min}^2}+\frac{1}{\f_{\min}})^{0.5}\cdot\frac{O(1)}{ \sqrt{T+2}}.$$

\end{compactitem}

\end{corollary}

\subsection{Repeating Analysis for Choosing $\vartheta(t) := - \log t$}

\begin{lemma}[Spectral Norm of Hessian $\vartheta(t)=-\log{t}$]
For $(x,y) \in \mcal{B}^{\geq}_\f \times \mcal{B}^{\leq}_\g$,

$\| \nabla^2 \Phi(x,y) \| = \max_{u \in V} (\frac{1}{x_u}+\frac{1}{y_u^2})
\leq \frac{1}{\g_{\min}}+\frac{1}{\f^2_{\min}}.$

\end{lemma}

\begin{proof}

$\Phi(x,y)=-\sum\limits_{u\in V}y_u\cdot \log{\frac{x_u}{y_u}}$

We consider the non-zero partial derivatives:

$$\frac{\partial{\Phi}(x, y)}{\partial{x_u}}=-{\frac{y_u}{x_u}}, \frac{\partial{\Phi}}{\partial{y_u}}(x,y)=1-\log{\frac{x_u}{y_u}}$$

$$
\frac{\partial^2{\Phi}(x,y)}{\partial{x^2_u}} = \frac{y_u}{x^2_u},
\frac{\partial^2{\Phi}(x, y)}{\partial{y^2_u}}=\frac{1}{y_u},
\frac{\partial^2{\Phi}(x,y)}{\partial{x_u}\partial{y_u}}=-\frac{1}{x_u}.$$

Note that the Hessian matrix has the form 
$\nabla \Phi(x,y) = \begin{pmatrix}
{D}_1&{D}_2\\
D_2&D_3
\end{pmatrix}$, where $D_i$ are $|V|\times |V|$ diagonal matrices for $i=1,2,3$. Then, the eigenvalues of the Hessian matrix are the roots of the polynomial in $\lambda$: 

$$\lambda^{|V|}\cdot\Pi_{u\in V}(\lambda-(\frac{1}{y_u}+\frac{y_u}{x^2_u}))=0.$$ 

Thus, we have $$\|\nabla^{2}\Phi(x,y)\|_2=\max_{u \in V} (\frac{1}{y_u}+\frac{1}{x_u^2}),$$

as required.

\end{proof}

\begin{lemma}[$\vartheta(t)=-\log{t}$]
For $\theta \in [0,1]$ and $u \in V$,
$\Gamma''(\theta) \geq 
\g_{\min}^3 \cdot \|\rho - \rho^*\|_2^2
\geq { \f^2_{\min}\cdot\g^3_{\min}} \cdot \epsilon_u^2$,
where $\epsilon_u := \frac{\rho_u - \rho^*_u}{\rho^*_u}$.
\end{lemma}

\begin{proof}
Denote $x(\theta) := x^* + \theta \cdot (x - x^*) \in \R^V$
and $y(\theta) := y^* + \theta \cdot (y - y^*) \in \R^V$.

Moreover, denote
$\rho(\theta) \in \R^V$ by
$\rho_u(\theta) := \frac{x_u(\theta)}{y_u(\theta)}$ for $u \in V$.
We have:

$\rho_u'(\theta) = y_u(\theta)^{-1} \cdot \{(x_u - x^*_u) - \rho_u(\theta) \cdot (y_u - y^*_u)\}.$

Then, we have:

$\Gamma(\theta) = -\sum_{u \in V} (y^{*}_u+\theta(y_u - y^*_u))\cdot \log{\rho_u(\theta)}$.

$\Gamma'(\theta) = -\sum_{u \in V} (y_u - y^*_u)\cdot \log{\rho_u(\theta)}+\rho^{-1}_{u}(\theta)\cdot\{(x_u - x^*_u)-(y_u - y^*_u)\cdot \rho_u(\theta)\}$.

Therefore, we have:

\begin{align*}
\Gamma''(\theta) & = \sum_{u \in V} \frac{y_u(\theta)}{x^2_u(\theta)} \cdot
\{(x_u - x^*_u) - \rho_u(\theta) \cdot (y_u - y^*_u)\}^2 \\
& \geq \frac{\g_{\min}}{\f_{\max}^2} \sum_{u \in V}  \{(x_u - x^*_u) - \rho_u(\theta) \cdot (y_u - y^*_u)\}^2.
\end{align*}

Denote $\epsilon_u(\theta) := \frac{\rho_u(\theta) - \rho^*_u}{\rho^*_u}$,
where $\epsilon_u = \epsilon_u(1)$.
Observe that $\epsilon_u \cdot \epsilon_u(\theta) \geq 0$ and
and $|\epsilon_u(\theta)| \leq |\epsilon_u|$.

Note that each term in the above sum contains the following expression:

$|(x_u - x^*_u) - \rho_u(\theta) \cdot (y_u - y^*_u)|
= |(\epsilon_u - \epsilon_u(\theta)) \cdot \rho^*_u y_u + \epsilon_u(\theta) \cdot \rho^*_u y^*_u|
\geq |\epsilon_u| \cdot \rho^*_u \cdot \g_{\min} \geq |\epsilon_u| \cdot \frac{\f_{\min} \cdot \g_{\min}}{\g_{\max}}.$

Therefore, we have:

$$\Gamma''(\theta) \geq 
 \frac{\g^3_{\min}}{\f_{\max}^2} \cdot\sum_{u \in V}  \epsilon_u^2 \cdot (\rho^*_u)^2
= \frac{\g^3_{\min}}{\f_{\max}^2} \cdot \|\rho - \rho^*\|_2^2
\geq \frac{ \f^2_{\min}\cdot\g^3_{\min}}{(\f_{\max}\cdot\g_{\max})^2} \cdot \epsilon_u^2,$$

for all $u \in V$, as required.

\end{proof}

\begin{corollary}[Approximation Guarantees for Density Vector $\vartheta(t)=-\log{t}$]
After $T$ steps of Frank-Wolfe,
the solution $z^{(T)} \in \mcal{B}^{\geq}_\f\times\mcal{B}^{\leq}_\g$
induces a density vector $\rho^{(T)} \in \R^V$ with the following approximation 
guarantees:

\begin{compactitem}

\item \emph{Absolute error}.

$$\| \rho^{(T)} - \rho^* \|_2 \leq \frac{1}{ \g^{1.5}_{\min}}\cdot(\frac{1}{\g_{\min}}+\frac{1}{\f^2_{\min}})^{0.5}\cdot\frac{O(1)}{ \sqrt{T+2}}.$$

\item \emph{Coordinate-wise Multiplicative error.}
For all $u \in V$,

$$\left|\frac{\rho^{(T)}_u - \rho^*_u}{\rho^*_u} \right| \leq \frac{1}{\f_{\min}\cdot \g^{1.5}_{\min}}\cdot(\frac{1}{\g_{\min}}+\frac{1}{\f^2_{\min}})^{0.5}\cdot\frac{O(1)}{ \sqrt{T+2}}.$$

\end{compactitem}

\end{corollary}

\section{Locally Maximin Solution as a Universal Minimizer for Data-Processing Divergences}
\label{sec:general_div}

In Section~\ref{sec:technical},
we see that given an instance $(V; \f, \g)$ with strictly monotone $\g$,
there is some distribution $\vecp \in \Delta(\mcal{S}_V)$
that induces a \emph{universal} pair $(\f^\vecp, \g^\vecp)$
that, for all convex $\vartheta$, attains the minimum
for:

$$\min_{(x,y) \in \mcal{B}^{\geq}_\f \times \mcal{B}^{\leq}_\g} \D_\vartheta(x \| y).$$

We show that the result can be generalized to 
a wider class of divergences that satisfy the data processing inequality.
As before, we assume that the functions $\f$ and $\g$ are normalized
to $\f(V) = \g(V) = 1$.

\begin{definition}[Divergences Satisfying Data Processing Inequality]
\label{defn:dpi}
A divergence measure is a function 
$\D$ that takes two distributions $P$ and $Q$ on the same space such that
$\D(P \| Q) \geq 0$, where equality holds \emph{iff} the distributions $P = Q$ are identical.

A divergence $\D$ satisfies the \emph{data processing inequality} if, for any stochastic transformation (or channel) $T$ that maps the original space to another space, the following inequality holds:

$$\D(T(P) \| T(Q)) \leq \D(P \| Q),$$

where $T(P)$ and $T(Q)$
are the resulting distributions after applying the transformation $T$
to $P$ and $Q$, respectively.

We use $\mathcal{D}_{\mathsf{DPI}}$ to denote the class of divergence that satisfies the data processing inequality.
\end{definition}

\begin{definition}[Universally Closest Pair]
\label{defn:universal}
Suppose $\mcal{A}$ and $\mcal{B}$ are collections of distributions
on the same space and $\mcal{D}$ is a class of divergences.
Then, a pair $(x, y) \in \mcal{A} \times \mcal{B}$ is universally closest 
from $\mcal{A}$ to $\mcal{B}$ with respect to $\mcal{D}$ if
for all $\D \in \mcal{D}$, the pair $(x,y)$ attains
the optimal:

$\min_{(x,y) \in \mcal{A} \times \mcal{B}} \D(x \| y)$.
\end{definition}

We generalize the above result in the following theorem.

\begin{theorem}[Locally Maximin Induces a Universally Closest Pair]
\label{th:universal_closest}
Given an instance $(V; \f, \g)$ with strictly monotone $\g$,
suppose the distribution $\vecp \in \Delta(\mcal{S}_V)$
is locally maximin as in Definition~\ref{defn:local_maximin}.
Then, 
the induced pair $(\f^\vecp, \g^\vecp)$
is universally closest from 
$\mcal{B}^{\geq}_\f$ to  $\mcal{B}^{\leq}_\g$
with respect to $\mcal{D}_{\mathsf{DPI}}$.
\end{theorem}

\noindent \textbf{Restricted Class of Hockey-Stick Divergences.}
As in~\cite{DBLP:conf/innovations/ChanX25},
the universally closest pair problem is first considered
on a more restricted class of divergences.

\begin{definition}[Hockey-Stick Divergence]
\label{hockeystick}
Given distributions $P$ and $Q$ on the same
sample space~$\Omega$ and $\gamma \geq 0$,
the hockey-stick divergence\footnote{Strictly speaking, $\mathsf{HS}_\gamma$
is a divergence only for $\gamma \geq 1$.
For $0 \leq \gamma < 1$, $\mathsf{HS}_\gamma(P \| P) = 1 - \gamma > 0$;
hence, to make it a divergence,
we have to add a constant and consider 
$\mathsf{HS}_\gamma(P \| Q) - 1 + \gamma \geq 0$,
where equality holds \emph{iff} $P = Q$.
However, it is still fine to consider $\mathsf{HS}_\gamma$ as an objective function.}
is defined as 

$\mathsf{HS}_{\gamma}(P \| Q) := \sup_{S \subseteq \Omega} P(S) - \gamma \cdot Q(S).$\\

Note that for finite $\Omega$,
by considering convex $\vartheta_\gamma(t) := \max\{t - \gamma, 0\}$,
we have:

$\mathsf{HS}_\gamma(P \| Q) = \D_{\vartheta_\gamma}(P \| Q)$.

We denote the class of hockey-stick divergences as $\mcal{D}_{\mathsf{HS}} := \{\mathsf{HS}_\gamma: \gamma \geq 0\}$.
\end{definition}

To prove Theorem~\ref{th:universal_closest},
we use the following fact
that is derived from the proof 
in~\cite[Theorem 5.12]{DBLP:conf/innovations/ChanX25}.

\begin{fact}[From Hockey-Stick to Data Processing Inequality]
\label{fact:HStoDPI}
Suppose $\mcal{A}$ and $\mcal{B}$ are collections of distributions
on the same space and 
a pair $(x, y) \in \mcal{A} \times \mcal{B}$ is universally closest 
from $\mcal{A}$ to $\mcal{B}$ with respect to $\mcal{D}_{\mathsf{HS}}$.
Then, $(x,y)$ is also universally closest
from $\mcal{A}$ to $\mcal{B}$ with respect to $\mcal{D}_{\mathsf{DPI}}$.
\end{fact}

The hypothesis in Fact~\ref{fact:HStoDPI}
can be paraphrased from Corollary~\ref{cor:local_maximin}
in our setting as follows.

\begin{lemma}[Paraphrasing Corollary~\ref{cor:local_maximin}]
\label{lemma:hs}
Any locally maximin~$\vecp$ induces a pair $(\f^\vecp, \g^\vecp)$
that is universally closest 
from $\mcal{B}^{\geq}_\f$ to  $\mcal{B}^{\leq}_\g$
with respect to $\mcal{D}_{\mathsf{HS}}$.
\end{lemma}

\begin{proof}
Observe that for any $\gamma \geq 0$,
the hockey-stick divergence $\mathsf{HS}_\gamma$
can be expressed as $\D_{\vartheta_\gamma}$ for some (not necessarily strictly) convex function~$\vartheta_\gamma$.  
\end{proof}

\begin{proofof}{Theorem~\ref{th:universal_closest}}
This is a direct application of Fact~\ref{fact:HStoDPI}.
The hypothesis of Fact~\ref{fact:HStoDPI} is achieved
by Lemma~\ref{lemma:hs}.
The conclusion of Fact~\ref{fact:HStoDPI} finishes the proof.
\end{proofof}

\newpage
\bibliographystyle{alpha}
\bibliography{ref,density,related work}


\end{document}